\documentclass[11pt,  oneside, a4paper]{article}
\usepackage[utf8]{inputenc}
\usepackage[a4paper, total={6.5in,9in}]{geometry}

\usepackage{graphicx} 
\usepackage{bbm}
\usepackage{amsfonts}
\usepackage{enumitem}
\usepackage{fullpage}
\usepackage{amsmath}
\usepackage[hidelinks]{hyperref}
\usepackage{mathtools}
\usepackage{lipsum}
\usepackage{blindtext}
\usepackage{titlesec}
\usepackage{amssymb}
\usepackage{amsthm}
\usepackage{rotating}
\usepackage{subcaption}
\usepackage{caption}
\usepackage{svg}
\usepackage{xcolor}
\usepackage{verbatim}
\usepackage{natbib}

\usepackage{balance}
\usepackage{breqn}

\theoremstyle{remark}
\newtheorem{remark}{Remark}

\newtheorem{claim}{Claim}

\newtheorem{defn}{Definition}
\newtheorem{thm}{Theorem}
\newtheorem{lem}{Lemma}

\newtheorem{model}{Model}
\newtheorem*{des}{Desiderata for Good Allocation Rules}
\newenvironment{proofsketch}{%
  
  \begin{proof}
}{
  \end{proof}
}

\newcommand{\footremember}[2]{%
    \footnote{#2}
    \newcounter{#1}
    \setcounter{#1}{\value{footnote}}%
}

\title{Fixed Points and Stochastic Meritocracies: \\A Long-Term Perspective}
\author{Gaurab Pokharel\footremember{alley}{Virginia Tech. Email: gaurab@vt.edu (Equal Contribution)} 
\and Diptangshu Sen\footremember{trailer}{Georgia Institute of Technology. Email: dsen30@gatech.edu (Equal Contribution)}
\and Sanmay Das\footremember{something}{Virginia Tech. Email: sanmay@vt.edu}
\and Juba Ziani\footremember{somethingelse}{Georgia Institute of Technology. Email: jziani3@gatech.edu}
}
\date{\today}

\begin{document}

\maketitle

\begin{abstract}
We study group fairness in the context of feedback loops induced by meritocratic selection into programs that themselves confer additional advantage, like college admissions. We introduce a stylized, yet novel inter-generational model for the setting and analyze it in situations where there are no underlying differences between two populations. When the benefit of the program (or the harm of not getting into it) is completely symmetric, we show that disparities between the two populations will vanish on average in the long term, although in the short term disparities will continue to arise and dissipate cyclically. Further, the time an accumulated advantage takes to dissipate can be significant, and increases as a function of the relative importance of the program in conveying benefits. Interestingly, significant disparities can arise purely due to randomness even from completely symmetric initial conditions, especially when populations are small. The introduction of even a slight asymmetry, where the group that has accumulated an advantage becomes slightly preferred, leads to a completely different outcome. In these instances, starting from completely symmetric initial conditions, disparities between groups arise stochastically and then persist over time, yielding a permanent advantage for one group. Our analysis precisely characterizes conditions under which disparities persist or diminish, with a particular focus on the role of the scarcity of available spots in the program and its effectiveness. We also present extensive simulations in a richer model that further support our theoretical results in the simpler, stylized model. Our findings are relevant for the design and implementation of algorithmic fairness interventions in similar selection processes.
\end{abstract}


\section{Introduction}\label{sec:intro}
Algorithms increasingly mediate critical resource allocation, influencing outcomes across education, healthcare, employment, and beyond \citep{chiara_gil_tully, crawford2019ainow, Kleinberg_Raghavan_2021}. In scenarios of inherent scarcity \citep{Investopedia2025Scarcity}, resources are often allocated according to principles of \emph{local justice} \citep{elster1992local}. While local justice principles are broad and can include prioritizing the most vulnerable or those most likely to benefit, our focus here is on so-called \emph{meritocratic} selection procedures. Within algorithmic fairness, meritocratic selection is closely aligned with the principle of \emph{individual fairness}, which requires similar individuals to receive similar outcomes \citep{Dwork_Hardt_Pitassi_Reingold_Zemel_2012}. Meritocratic fairness specifically dictates that similarly qualified individuals obtain comparable outcomes, and those with higher qualifications take priority—for instance, in admissions to prestigious colleges or hiring for high-skilled jobs.

While the ideal of meritocracy is to reward talent over social background \citep{Young1958Rise}, critics argue it can become a ``myth'' that reinforces inequality \citep{McNamee_Miller_2009, Castilla2010Paradox, Marvel_2025}, particularly when `merit' is operationalized through flawed, class-bound proxies. For instance, in higher education, using test scores and ``race-blind'' rankings as the primary measure of academic merit can systematically disadvantage students with fewer resources and reduce diversity \citep{Chetty_Deming_Friedman_2023, kizilcec2023holistic}. This reliance on crude proxies often contrasts with the more holistic, equitable judgments of human experts \citep{pokharel2024discretionary}. This pattern extends across domains, with algorithms perpetuating historical biases in hiring by misinterpreting `job fit' \citep{raghavan2020hiring}, misjudging `need' in welfare allocation \citep{eubanks2018automating}, and using biased `crime risk' scores to create cycles of discriminatory surveillance \citep{LumIssac2016Serve}. Under perfect measurement of qualifications, meritocratic selection appears inherently fair in static settings; however, because allocation itself confers lasting advantages, fairness can become complicated over longer timescales. Our work thus explores the critical question: do resource allocation processes that are both meritocratic and fair in static settings necessarily remain fair over time, or might they generate persistent unfairness? If yes, under what conditions does it emerge?

We introduce a stylized theoretical framework, designed to be as parsimonious as possible, to systematically explore these dynamics. The essence of the basic model is that individuals can belong to one of two types (that we call ``high'' and ``low''), as well as one of two distinct groups. High-type individuals are prioritized for college admission\footnote{%
While we use the running example of college admission as a representative case to simplify the exposition, we note that this can apply to a broad range of meritocratic selection processes including in employment, sports, and other education contexts.}, but if the college capacity is large enough, others may also be admitted. Access to college confers a high probability of remaining (or becoming) a high type, while not being admitted means an individual has a significantly lower probability of remaining (or becoming) so. After college, we assume the population replicates to the next generation, and the process repeats. By design, the model is completely symmetric in how it treats the two groups.

We are interested in the dynamics governing how both the proportion of high-type individuals and college admissions outcomes evolve over generations. We show that under our basic model, \emph{separations} where one group has a higher proportion of high types than the other, eventually vanish, as the system repeatedly converges back toward equitable fixed points. Group outcomes undergo cycles: stochastic forces periodically push groups apart, and the system dynamics subsequently restore fairness. However, it is worth noting that these stochastic pushes that cause separations can be significant, especially when populations are smaller in size. An important additional note is that if the groups start off in a separated state (for example, because of long-term effects of past discrimination), the system can take time to converge back to an equitable state. The duration of this unfair period depends on both college selectivity and the relative efficacy benefit gained from attending college. When college attendance provides a strong support for maintaining advantaged status, existing disparities persist longer, extending periods of inequality and potentially amplifying their impact.

Next, we ask what happens when we introduce a slight asymmetry. Our second model introduces a so-called \textbf{affinity advantage}. In simple terms, when one group gains a slight lead in how many of its members are high-type---such as securing more positions in college admissions or other valuable resources---this group's other members also gain incremental boosts in their likelihood of becoming or remaining advantaged. There are many possible mechanisms that could lead to such an advantage -- for example, members of a currently advantaged group benefiting from enhanced visibility, greater support networks stemming from in-group preferences \citep{in_group, everett2015InGroup, boat2022SocialCapital}, or increased opportunities \citep{seibert2001CareerSuccess, Austin2023}. 
In our model, this translates concretely to a slight increase in the transition probability towards a high-type status for the members of the currently leading group. Our theoretical analysis demonstrates that even \emph{minimal} advantages created through this mechanism can become self-reinforcing, leading to stable and substantial long-term disparities. Crucially, these differences are not due to inherent asymmetries; rather, they emerge inevitably once stochastic fluctuations give rise to small, initial separations between otherwise identical groups. This offers an insight into an important empirical finding in the social science literature on how  small initial advantages tend to grow over time, a phenomenon termed the ``Matthew Effect'' \citep{Merton_1988, Stanovich_1986}. This dynamic of cumulative advantage has been observed in contexts ranging from educational attainment \citep{Bloome2018} to career success \citep{Nelson_Wilmers_Zhang_2025, seibert2001CareerSuccess}.

While our model is, by design, extremely parsimonious, it is able to capture the essential dynamics. We confirm this by conducting simulation experiments in a richer environment model, where ``merit scores'' are real-valued (rather than everyone just belonging to one of two types), as is the stochasticity arising from college advantage and inter-generational replication. These simulations confirm that the phenomena uncovered in our theoretical models are robust and not mere artifacts of simplifying assumptions. 

Our overall findings highlight limitations in static fairness frameworks commonly employed today, underscoring the need to integrate temporal and stochastic considerations into fairness evaluations \citep{Liu2018Delayed, Mouzannar2019, xu2024equal}. Our results carry implications for algorithm designers, policymakers, and educational institutions, emphasizing the necessity of proactive, dynamic interventions to prevent minor, early-stage disparities from crystallizing into lasting inequalities.

\paragraph{Paper Structure} The paper structure is as follows:
\begin{enumerate}[label=\alph*)]
    \item In Section~\ref{sec:model}, we introduce our stylized framework of fair and meritocratic selection into a scarce program (henceforth ``college'' for expositional purposes). 
    In our framework, college is a scarce resource that provides all individuals who attend a high chance of upward social mobility. For individuals who cannot attend college, we first consider our \textbf{Equal Advantage} (EA) model---where no group has any preferential advantage---, and then the \textbf{Affinity Advantage} (AA) model---where the group with more high-type individuals gets a small affinity advantage due to network effects. Importantly, our model is parsimonious \textit{by design}, with the goal of isolating the role of randomness and network effects as drivers of long-term unfairness in the simplest possible model. 
    \item In Section~\ref{sec:results}, we analyze both short-term and long-term behaviors of the system under Equal Advantage. While groups achieve parity in the long term \textit{in an average sense}, in the short term, randomness spontaneously drives the groups to separate before the separation dissipates again. Further, and importantly, this happens \textit{cyclically}. We then move on to the Affinity Advantage model and show that groups can actually achieve persistent separation even in the long-term when one group has a very small affinity advantage. Our results highlight that unfairness can and will arise even in the simplest settings with initially symmetric populations and fair decisions at every step. These forces only get exacerbated in the real world, amplifying the effects our model predicts. 
    \item Finally, in Section~\ref{sec:full_model} we describe our richer simulation model that relaxes the binary assumptions from our theoretical framework in Sections~\ref{sec:model} and ~\ref{sec:results}. Through extensive numerical experiments on this simulation model, we demonstrate the robustness of our insights about small affinity advantages driving long-term group disparities. 
\end{enumerate}

Our work is closely aligned with the literature on dynamic fairness and inter-generational feedback loops (we include a detailed review of this literature in Section~\ref{sec:background}). Unlike this literature which \textit{always} assumes that unfairness arises due to some form of initial disparities between groups (distributions of qualifications; update rules after going through college; costs to take actions; etc.) or unfair decision making, \textit{ours} is the first work to formalize how simply \textit{randomness} across initially symmetric populations can still lead to repeated reappearance of unfairness, \textit{even if the decisions made at each round are fair} (as in our EA model). This has further implications: reaching demographic parity is not enough for long-term fairness, as unfairness can and will always re-appear—a significant departure from the typical view that there is an “end goal” to fairness. We also consider simple yet novel network effects (in our AA model) that were not previously studied in the context of fairness, and highlight how they amplify unfair outcomes in the long term.

\section{Background and related work}\label{sec:background}

The proliferation of algorithmic decision-making has spurred a vast body of research on fairness \citep{Dwork_Hardt_Pitassi_Reingold_Zemel_2012, friedlerImpossibility2021, Kleinberg2016Inherent, barocas2023FairnessandML, Barocas_Selbst_2016, crawford2019ainow}. Much of this foundational work has focused on developing and auditing static fairness metrics, which assess outcomes at a single point in time. These metrics range from group-level constraints like demographic parity and equality of opportunity \citep{hardt2016equality} to individual-level requirements that similar individuals be treated similarly \citep{Dwork_Hardt_Pitassi_Reingold_Zemel_2012}. While crucial, these static snapshots often fail to capture the long-term effects of decisions \citep{static_notions, corbett2018measure}. 

\paragraph{Meritocratic Notions of Fairness.} Our work adopts a \emph{meritocratic} selection principle---a common approach in both policy and fairness literature where individuals with higher measured qualifications are prioritized \citep{elster1992local, joseph2016fair}. Though intuitively appealing, the notion of meritocracy has been critically examined. Studies show that when merit is defined by flawed or biased proxies---such as standardized test scores for college admission or healthcare costs for medical need—meritocratic systems can reinforce existing inequalities rather than remedy them \citep{McNamee_Miller_2009, Castilla2010Paradox, obermeyer2019dissecting, Chetty_Deming_Friedman_2023}. While this prior work critiques the implementation of meritocracy, we use it as a starting point to investigate its downstream consequences when qualifications are measured perfectly. The main novelty of our work on this front is to understand how static meritocratic notions of fairness applied at each time step impact fairness in the long-term, and to argue that static interventions may not be sufficient in practice.
 
Our work is closely related to \citet{zhang2020LongTerm}, who analyze long-run group qualification under threshold policies within a POMDP framework and study how imposing static group or opportunity-based fairness constraints (e.g., demographic parity or equalized odds) affects equilibrium outcomes. In contrast, we (i) assume individually fair, meritocratic selection without statistical constraints, (ii) model scarcity via a fixed capacity $\alpha$, and (iii) show how disparities can \emph{emerge from symmetry} due solely to stochasticity and then become entrenched through a minimal, endogenous affinity advantage.

\paragraph{Long-term Dynamics and Fairness.} Because of such limitations of static analysis, there has been a growing focus on dynamic fairness and feedback loops, where the outcomes of today's decisions influence the inputs of tomorrow. The foundational work of \citet{Liu2018Delayed} showed that enforcing a static fairness metric can paradoxically harm a group's long-term outcomes. This insight has fueled research into models of runaway feedback in predictive policing \citep{Ensign2018Policing}, cumulative advantage in social systems \citep{Merton_1988, Stanovich_1986}, and inter-generational mobility in resource allocation \citep{HeidariKleinberg2021Intergenerational, acharya2023wealth, HuChen2018, Mouzannar2019}. 
Most of this literature explores policy interventions in both the short- and long-terms (such as wealth redistribution, affirmative action etc) to counteract dynamics that amplify disparities, often with varying fairness goals~\cite{rateike2024designing,puranik2024long,zezulka2024fair}. In contrast, we study a different problem: identifying new factors (randomness and network effects) that can drive long-term unfairness even when decisions in the short term are fair, efficient and meritocratic.

\paragraph{Performative Prediction.} Our work is also related to the idea of performative prediction 
\citep{PerdomoPerformative2020, NarangPerformative2023, ShavitCausal2020, YahavGamingHelps2021, BrownHodKalemajPerformative2022, MillerOutsideEcho2021, ZrnicWhoLeads2021} (to only name a few references), i.e., when the model or decision rule deployed by a decision-maker (or multiple decision-makers playing a game against each other) at a given time step affects the outcomes and the data collected at this time step, inducing distribution shifts that can potentially alter dynamics in the long-term. Such performative prediction, sometimes due to strategic considerations, is known to have important implications for fairness 
\citep{MilliSocialCost2019,  HuDisparateEffects2019, SomerstepFiarnessPerformativePolicy2024, mishler2022performative, ZezulkaGenin2023, XieZhang2024, Bechavod2022InformationDiscrepancy, Avasarala_Wang_Ziani_2025}.

\paragraph{Randomness and Stochasticity in Population Dynamics.} Recently, the role of stochasticity has emerged as a critical factor in fairness. Research has shown that randomness inherent in machine learning pipelines—such as from data sampling or model initialization---can lead to significant and unpredictable variations in fairness outcomes for different groups \citep{ganesh2023variance}. Other work has explored the fundamental role of luck versus talent in achieving success \citep{talent_v_luck} and has analyzed the use of intentional randomness, like lotteries, in allocation systems \citep{Jain2024, Baker_Bastedo_2022, lee2024admissions}, showing that its effects on equity are complex and context-dependent.

Our work synthesizes and extends these threads. 
While prior research on dynamic fairness, like that of~\citet{baek2023feedback}, often assumes pre-existing asymmetries between groups (in terms of one or more of the following factors : i) distribution of types, ii) access to opportunities, and iii) evaluation), the main novelty of our work is that we investigate how inequality can emerge and persist between groups \emph{that are initially identical in all aspects}.
We build a unified framework to analyze the joint effects of meritocratic selection, feedback loops, and scarcity. Our primary contribution is to isolate the role of pure stochasticity as a mechanism for generating initial, temporary disparities, and then to demonstrate how these random fluctuations can be captured and permanently entrenched by even a minimal network-style ``affinity advantage''. By doing so, we formalize a pathway by which transient luck can crystallize into permanent social stratification.

\section{Model}\label{sec:model}

We consider a population of size $(2N)$ divided into two groups $A$ and $B$, each with exactly $N$ individuals. Each individual can be one of two types: \textit{low type} or \textit{high type}. The type of the individual is a proxy for their potential to be successful (for example, success can be measured in terms of one's `earning potential'). The population has access to a resource, which we call \textit{college}, which provides individuals with upward social mobility. However, college is a scarce resource and only a fraction\footnote{Our objective is to understand the effects of meritocratic selection to scarce resources like college, so we restrict college capacity $\alpha$ to be less than $\frac{1}{2}$ of the total population.} $\alpha$ ($0 < \alpha < \frac{1}{2}$) of the total population who meet the \textit{college selection rule} can avail it. The college capacity is denoted by $C$ where $C = (2N) \alpha$. Our goal is to understand how preferential access to college and other social dynamics affect the evolution of `successful' or \textit{high-type} individuals across different groups in a discrete-time setting. At each time step $t \in \{0,1,2,...\}$, we measure the fraction of high-type individuals $X_A(t)$ and $X_B(t)$ in groups $A$ and $B$ respectively. Clearly, $0 \leq X_A(t), X_B(t) \leq 1$ for all $t \geq 0$. We also define $X(t) = X_A(t) + X_B(t)$. Note that $(X_A(t), X_B(t))$ completely characterizes the state of the system at time $t$.  

\begin{remark}
We deliberately focus on the special case with equal group sizes for simplicity of exposition and to isolate randomness and network effects as drivers of long-term unfairness, separating from population size effects. However, all our results extend to the setting where group sizes are different. See Appendix~\ref{app:unequal} for details.
\end{remark}

\subsection*{Selection Rule} The mix of students admitted to college at any time step depends on the state of the system in the previous step. For $i \in \{A, B\}$, let $A_i(t+1)$ be the number of students from group $i$ admitted to college at time $t+1$. We want our college selection rule\footnote{We use the terms selection rule and allocation rule interchangeably.} to be fully deterministic --- this is, in many ways, the \emph{best case} from the perspective of fairness because it helps to isolate all uncertainty in the system only to outcomes. But, what constitutes a \textit{good} allocation rule?
\begin{des}
A \textit{good} allocation rule must satisfy all the following properties: 
\begin{itemize}
    \item it must be \textbf{meritocratic}, i.e., it always prioritizes a high type over a low type candidate for an available position, as long as any high type candidates remain to be allocated;
    \item it must be \textbf{fair}, in that it achieves equalized selection rates across groups for all candidate types (this is equivalent to the fairness notion of Equalized Odds introduced by ~\citet{hardt2016equality}); and
    \item finally, it must be \textbf{efficient}, i.e., all available seats are allocated. 
\end{itemize}
\end{des}
We now describe our college allocation rule $\mathcal{A}$. We deliberately make a distinction between scenarios when college capacity is scarce (not all high types can go to college) versus when capacity is abundant (all high types in the population can go to college) because it slightly alters the allocation as we see below:

\begin{defn}[Allocation Rule $\mathcal{A}$]\label{def:allocation_rule}
When capacity is limited (\textit{over-subscribed regime}) with $NX(t) \geq C$, no low type individuals in the population get admitted in round $t+1$. All available seats are allocated exclusively to high type individuals, with each group receiving seats proportional to the fraction of high types in their respective populations in round $t$. Therefore, the total number of admits for each group $i$ is given by: 
\[
     A_i(t+1) = \frac{X_i(t)}{X(t)}\cdot (2N\alpha) \quad \forall~i. 
\]
However, when capacity is abundant (\textit{under-subscribed regime}) with $NX(t) < C$, the college first admits all high type individuals from both groups. The residual capacity is then filled by low type individuals, with each group receiving seats proportional to the fraction of low types in their respective populations in round $t$. In this case, the total number of admits for each group $i$ is given by: 
\[
     A_i(t+1) = NX_i(t) + \frac{1-X_i(t)}{2-X(t)}\cdot \left(C - NX(t)\right) \quad \forall~i.
\]
\end{defn}

\noindent 
Our first result shows that not only is our allocation rule natural, it is, in fact, the \emph{only} rule satisfying our requirements for what it means for a deterministic rule to be a good one in this setting. 
 \begin{thm}\label{thm:allocation}
Allocation rule $\mathcal{A}$ is the only deterministic allocation rule which is simultaneously meritocratic, fair, and efficient. 
\end{thm}

\begin{proofsketch} 
The proof is constructive in nature. We impose the three desirable properties outlined above and try to construct compliant allocation rules. We find that there is actually a unique deterministic rule which satisfies all three properties simultaneously and it is our selection rule $\mathcal{A}$. For a detailed proof, please refer to Appendix~\ref{app:theory}. 
\end{proofsketch}
\begin{remark}
For unequal group population sizes, we can again use the same principles as above to reconstruct the unique allocation rule for that setting. The general form of the allocation rule is presented in Appendix~\ref{app:unequal}. Note that it reduces exactly to $\mathcal{A}$ for the equal group size setting. All of our main insights go through unchanged. 
\end{remark}

\subsection*{Type Transitions} At every time step $t$, whether an individual transitions to high type or low type in the next time step depends primarily on whether they get admitted to college or not. 

\paragraph{Admitted Students} Each individual that goes to college transitions to a high type (independently of others) in the next time step with probability $p$, irrespective of their starting type or group membership\footnote{To isolate disparities driven by randomness and feedback effects, we assume that the groups are intrinsically identical, so college affects them identically.}. Thus, $p$ can be perceived as the probability of success from going to college. 

\paragraph{Rejected Students} The students who do not go to college can transition according to one of two models. The first model is called the ``Equal Advantage'' model, and considers situations where being part of the majority high-type group does not confer advantages to low-type agents: 

\begin{model}[\textbf{Equal Advantage}]
In the first model, any high type individual (irrespective of their group membership) who does not go to college, can still retain their type, albeit with a much smaller probability $q$ ($q \ll p$) \footnote{In our model, college provides a significant advantage, making it the strictly preferred option for all individuals, regardless of type.}.
However, any low type individual (across both groups) who does not go to college remains a low type with probability $1$. We call this the \textit{Equal Advantage} model because both groups face exactly identical dynamics over generations with no preferential treatment at any juncture whatsoever. 
\end{model}

In our second model, called the ``Affinity Advantage'' model, we introduce a new modeling element: namely, all individuals that do \emph{not} go to college benefit from being in a group with a larger number of high types. This may be due to network efforts, where for example, individuals of a given group prefer hiring individuals from the same group, or where low-type individuals directly benefit from connections with and exposure to high types: 

\begin{model}[\textbf{Affinity Advantage}]
This model retains all the characteristics of the first model, but there is a key distinction. If group $i$ has more high type individuals at time $t$, all group $i$ non-college-goers in the next round obtain an \textit{affinity advantage} $\epsilon$ (in their chance of transitioning to high type) which does not extend to individuals of the other group $j \neq i$. This means that group $i$ non-college-goers transition as follows: high types remain high types with probability $q + \epsilon$ and low types have a chance to become high types with probability $\epsilon$. In contrast, in group $j \neq i$ and among individuals rejected from college, high types retain their types only with probability $q$ while low types remain low types with probability $1$ (like in the first model). The affinity advantage can be perceived as a network effect --- individuals who belong to the more successful group have a higher chance of success even if they do not end up in college. Such network effects have been studied extensively in the social science literature \citep{Merton_1988, Stanovich_1986, Bloome2018, acharya2023wealth, pedulla2019race, wapman2022quantifying}. 
Two prominent real-world examples include: i) racial inequities perpetuated in labor market outcomes through network access (for example,~\citet{pedulla2019race} document how Black Americans, who have access to smaller networks facilitating employment, also apply and get hired at smaller rates); and ii) US universities hiring faculty mostly from a few elite universities due to network-based reputation effects~\cite{wapman2022quantifying}.
\end{model}

\section{Main Results}\label{sec:results}

Our goal is to compare how $X_A(t)$ and $X_B(t)$ evolve across the two transition models, particularly exploring if and when disparities can arise between the groups. Although $\{ \left(X_A(t), X_B(t) \right) \}$ is a stochastic process (which we will analyze in some detail), we will also consider an abstraction which will help us reason about the \textit{long-run average behavior} of the system. Before we present the main results of this section, let us introduce some additional notation and definitions.  

Given $(x, y) \in [0,1]^2$, define: 
\begin{align}\label{eqn:T}
      \mathcal{T}(x, y) = \left( \mathbb{E}_{|(x,y)} X_A(t+1), \mathbb{E}_{|(x,y)}X_B(t+1) \right),
\end{align}
where $\mathbb{E}_{|(x,y)} (\cdot) = \mathbb{E}[(\cdot)~|~X_A(t) = x, X_B(t) = y]$. $\mathcal{T}(\cdot)$ is a mapping that captures all the model-specific dynamics we have described earlier, but \textit{in expectation}. Given a starting point $(X_A(0), X_B(0))$, $\mathcal{T}(\cdot)$ thus defines a fully deterministic sequence of points $(x,y) \in [0,1]^2$. We can now define the \textit{fixed point} of the system as follows:

\begin{defn}[Fixed Point]\label{defn:fp}
$(x_A, x_B)$ is a fixed point of the system if and only if $(x_A, x_B)$ is a fixed point of $\mathcal{T}$, i.e.,
\[
          \mathcal{T}(x_A, x_B) = (x_A, x_B),
\]
where $\mathcal{T}(\cdot)$ is defined as in Equation~\eqref{eqn:T}. 
\end{defn}
A fixed point $(x_A, x_B)$ of the system can be interpreted as the \textit{expected fraction of high type individuals} in groups $A$ and $B$ once the system has \textit{stabilized}. We study these fixed points in greater detail for the EA and AA models.

\subsection{Equal Advantage}\label{sub:equal_adv}
We first consider the \textbf{Equal Advantage} transition model. In this case, we can derive the analytical form of $\mathcal{T}(\cdot)$ in closed form for both over-subscribed and under-subscribed regimes. 
This implies, by Definition~\ref{defn:fp}, that $(x_A, x_B)$ is a fixed point of the system:
\begin{itemize}
    \item under the \textit{over-subscribed regime} if and only if: 
    \begin{align*}
        x_A &= \frac{x_A(2\alpha)}{x_A + x_B} \cdot p + \left(x_A - \frac{x_A (2\alpha)}{x_A+x_B} \right)\cdot q,\\
        x_B &= \frac{x_B(2\alpha)}{x_A + x_B}\cdot p + \left(x_B - \frac{x_B (2\alpha)}{x_A+x_B} \right)\cdot q.
    \end{align*}
    \item under the \textit{under-subscribed regime} if and only if:
    \begin{align*}
        x_A &= x_A \cdot p + \frac{(1-x_A)(2\alpha - x_A - x_B)}{(2-x_A-x_B)} \cdot p, \\
        x_B &= x_B \cdot p + \frac{(1-x_B)(2\alpha - x_A - x_B)}{(2-x_A-x_B)} \cdot p.
    \end{align*}
\end{itemize}

The main results of this segment are organized as follows: we first present Theorem~\ref{thm:equal_adv} which highlights the long-term behavior of the system in an average sense using fixed points. Although the groups achieve parity in the long term, in the short term, the system is found to behave differently --- in particular, randomness can cause the groups to spontaneously separate (Figure~\ref{fig:separation_N}) before the dynamics inevitably force a return to parity (Theorem~\ref{thm:time_parity}). We also characterize how the degree of separation achieved and the time to achieve parity vary depending on the key system parameters.

\begin{thm}\label{thm:equal_adv}
Under the \textbf{Equal Advantage} model, the system has a unique fixed point $(x_A, x_B)$ which occurs in the under-subscribed regime (i.e., $x_A + x_B < 2\alpha$) and can always be reached from any starting point $(x, y) \in [0,1]^2$. Further, at the fixed point, we have $x_A = x_B = \alpha p$ implying that both groups achieve parity. 
\end{thm}

For a detailed proof of the above result, please see Appendix~\ref{app:theory}. Theorem \ref{thm:equal_adv} shows that under the Equal Advantage model, there is no disparate impact on either group in the long-term and both groups eventually reach parity. We demonstrate examples of this kind of behavior in Figure~\ref{fig:ea_parity} for different starting points of the system (both over-subscribed and under-subscribed regimes). This looks promising from a fairness perspective; but what happens in the short-term when there is stochasticity? We study this question in the next segment, showing that \textit{stochasticity can actually cause the system to behave differently on shorter time-scales}. \\

\begin{figure*}[ht]
  \centering
    \includegraphics[width=0.75\textwidth]{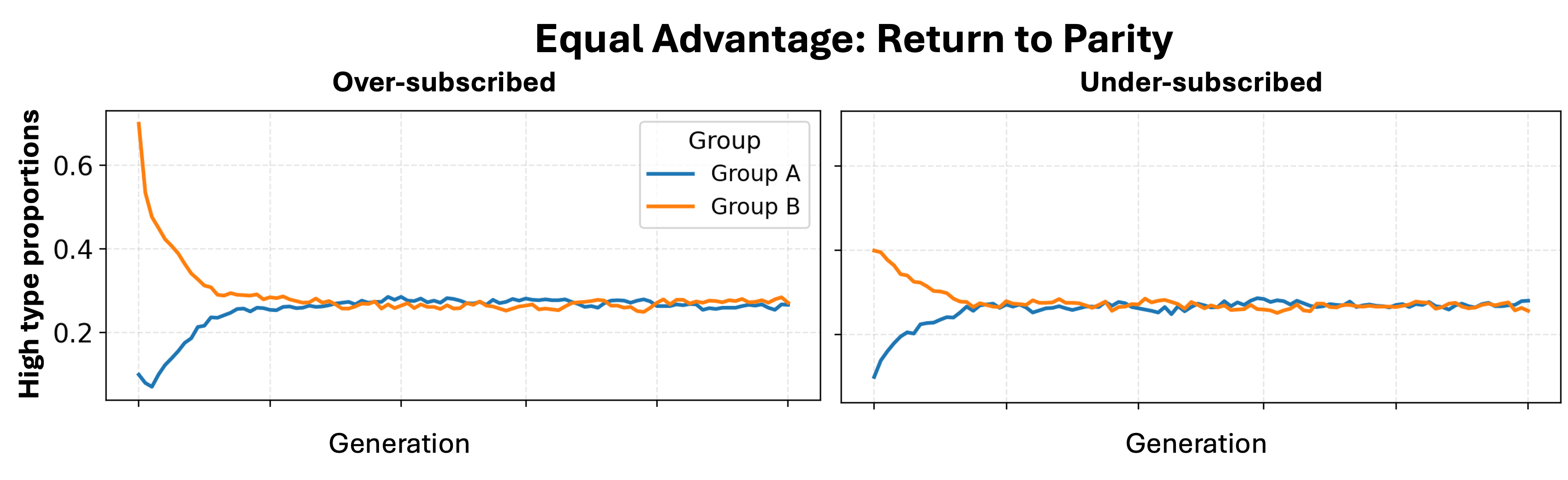}
    \caption{Sample trajectories of the fraction of high-types in each group over time. Left (over-subscribed, $(X_A(0),X_B(0))=(0.1,0.7)$) and right (under-subscribed, $(X_A(0),X_B(0))=(0.1,0.4)$) panels illustrate different starting points for the system. Parameter combination for runs: $p=0.90$, $q=0.40$, $\alpha=0.30$. In both cases, the system converges to $x_A=x_B=\alpha p=0.27$, which implies each group occupies half the college seats in the long run.}
    \label{fig:ea_parity}
\end{figure*}

\textbf{Separation.} Although Theorem~\ref{thm:equal_adv} shows that the system must reach parity on average in the long-run, it will not remain at parity throughout. Rather, stochasticity can sometimes drive the system to separate to some degree and then return to parity, before separating again. We measure separation between the groups at time $t$ as the absolute gap in the fraction of high-type individuals between the groups given by $|\Delta(t)| = |X_A(t) - X_B(t)|$. 

When it comes to separation, we are interested in how the extent of separation depends on our system parameters. In fact, the primary object of interest is the dependence on the population size $N$. When $N$ is small, stochasticity can drive large separation between groups in the short-term despite the initial populations being fully identical. Unsurprisingly, as $N$ grows large, stochasticity gradually becomes insignificant to the behavior of aggregate population statistics. We see this phenomenon clearly in Figure~\ref{fig:separation_N}.

\begin{figure*}[ht]
    \centering
    \includegraphics[width=0.65\textwidth]{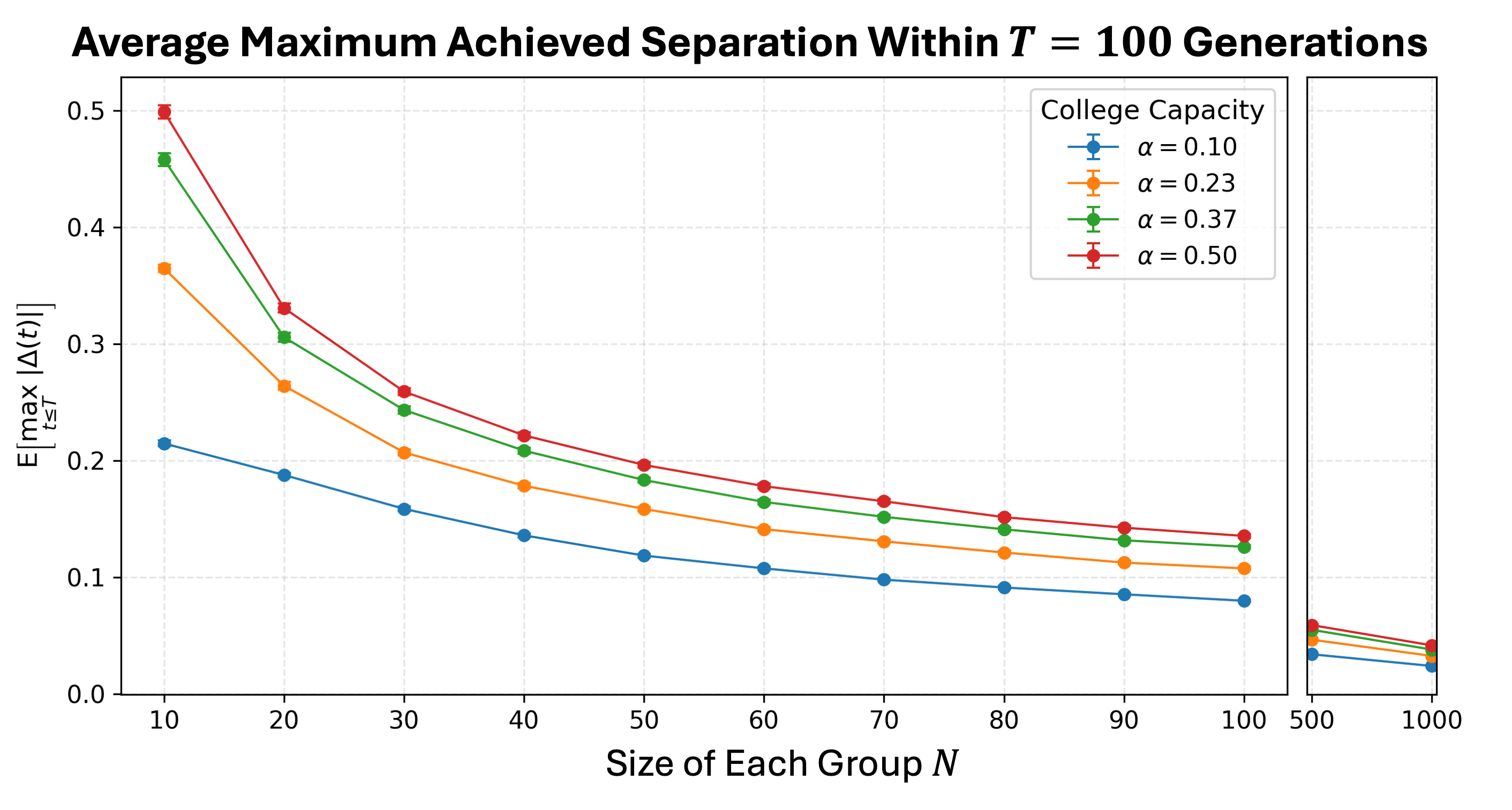}
    \caption{Average maximum separation $\mathbb{E}[\max_{t\le T}|\Delta(t)|]$ achieved versus population size $N$ for different capacity parameters $\alpha$ ($T=100, p=0.90, q=0.40$). The left panel shows the trend for smaller $N$ ($10 \le N \le 100$). The right panel highlights that the trend continues for large $N$ --- the expected maximum level of separation achieved over $T$ steps continues to shrink and approaches zero. At low capacities, the extent of separation is significant even for moderate $N$. }
    \label{fig:separation_N}
\end{figure*}

Further, note that separation is only possible to the extent that there can be high types in the total population, which is largely controlled by $\alpha$. Therefore, the extent of separation always has to be contextualized in terms of the allocation capacity parameter $\alpha$. One useful metric is the extent of separation normalized by $2\alpha$. At low capacities, this ratio can be very high which signals large disparities between the groups arising purely out of stochasticity. For example, at $\alpha = 0.1$, while small $N~(\approx 10)$ can lead to extreme disparity (the ratio approaches $1$ indicating that all high types in the population come exclusively from one group). Even for moderately large $N$ ($\approx 100$), we still observe significant disparities (ratio evaluates close to $0.5$). \textit{Thus, limited capacity can exacerbate the effect of stochasticity in driving differences between groups.} However, the separation effect diminishes with increasing capacity (for example, at $\alpha = 0.5$, the ratio evaluates to less than $0.1$), as increased capacity makes it easier for high types from both groups to get into college.    

\begin{remark}[Applications to real-world scenarios]
Our main running example is university admissions, where populations are typically large and stochasticity \emph{alone} plays a limited role in promoting unfairness. However, there are many practical settings where $N$ is relatively small and chance events can have outsized, lasting effects. Examples include:
\begin{itemize}
    \item Employees competing for promotion. Early random successes (such as securing a major grant) can drive large differences in career trajectories by enabling faster growth, greater visibility, and compounding advantages;
    \item Prestigious fellowships or internships with only a handful of recipients. Small early differences can yield disproportionate opportunities and long-term advantages in the job market;
    \item Age effects in youth athletics. Older children within an age group receive more attention and coaching simply because of earlier physical development, leading to greater success and chances of turning professional.\\
\end{itemize}
\end{remark}

\textbf{Return to Parity.} As per Theorem~\ref{thm:equal_adv}, the system eventually returns to parity, irrespective of how far apart the groups start. But, \textit{given that the system starts off at some (potentially large) separation $|\Delta(0)|$, how long would it take for the system to achieve parity ? }

Our next result answers this question, demonstrating how the time to reach parity depends on system parameters. We start by defining the notion of $\eta$-parity for $\eta > 0$:
\begin{defn}[$\eta$-parity]
For $\eta > 0$, we say that the system is at $\eta$-parity at time $t$ if the separation between groups at time $t$ does not exceed $\eta$, i.e., $|\Delta(t)| \leq \eta$. 
\end{defn}

\begin{thm}\label{thm:time_parity}
Given an initial separation $|\Delta(0)|$ and target $\eta > 0$ with $|\Delta(0)| > \eta$, the system will reach $\eta$-parity with high probability by time $T_{\eta}$. I.e., for $t \geq T_{\eta}$, we have that $\mathbb{P}\left[~|\Delta(t)| \leq \eta \right] \geq 1-\omega$ where 
\[
    T_{\eta} = \left \lceil  \log\left(  \frac{\omega - \frac{f(\alpha,p,q)}{N\eta^2 (1-p)} }{ \frac{|\Delta(0)|}{\eta} - \frac{f(\alpha,p,q)}{N\eta^2 (1-p)} }   \right)/ \log(p)     \right \rceil - 1,
\]
and $f(\alpha, p, q) = \alpha(1-p) + \frac{(1-\alpha)q(1-q)}{p}$. 
\end{thm}

Detailed proof can be found in Appendix~\ref{app:theory}. It is important to note that the above expression is meaningful only when for any given $\eta > 0$ and failure probability $\omega > 0$, the population size $N$ is $\Omega(\omega^{-1}\eta^{-2})$. When $N$ is too small, the variability of the system is too high, which impedes reliable convergence behavior. Additionally, we have $\frac{|\Delta(0)|}{\eta} > 1 > \omega$. This means that the numerator term in $\left \lceil \cdot \right \rceil$ is negative which would imply that a higher value of $|\Delta(0)|$ increases the magnitude of the numerator (the negative sign is offset by the negative sign from $\log(p)$) and thus leads to a larger $T_{\eta}$. This is a sanity check---starting off at a higher separation, the system \emph{should} take longer to return to parity. 

Our expression of $T_{\eta}$ predicts that the time to reach parity should increase with $p$ (the ``efficacy of college'' in yielding high type outcomes). This is what we observe empirically in Figure~\ref{fig:parity_p}, Appendix~\ref{app:exp}. The mean number of generations to reach $\eta$-parity given some initial separation $|\Delta(0)|$ is found to increase sharply with the success probability $p$. The intuition is the following: as $p$ approaches $1$, college outcomes tend to become more deterministic. This implies that any separation that has already been achieved becomes harder to mitigate. On the other hand, while our expression of $T_{\eta}$ does depend on $q$ (it is increasing for $q \in [0, \frac{1}{2}]$), empirically we find that $q$ has a negligible effect on the mean time to return to parity (Figure~\ref{fig:parity_q}). We believe that this phenomenon is driven by the fact that $q$ affects the dynamics only for a limited time (until the system is in the over-subscribed regime, if at all). Since the fixed point is in the under-subscribed regime, the system has to reach this regime at some point, beyond which $q$ is irrelevant.

\subsection{Affinity Advantage}\label{sub:affinity_adv}
We now consider the \textbf{Affinity Advantage} (AA) transition model. To simplify the exposition, we set $q = 0$ (although note that $q > 0$ only amplifies the separation effects introduced by $\epsilon$, as per Appendix~\ref{app:exp}). We can show that $(x_A, x_B)$ is a fixed point of the system:
\begin{itemize}
    \item under the \textit{over-subscribed regime} if and only if: 
    \begin{align*}
        x_A = \frac{x_A(2\alpha)}{x_A + x_B} \cdot p + \left(1- \frac{x_A(2\alpha)}{x_A + x_B} \right)\cdot \epsilon, \quad
        x_B = \frac{x_B(2\alpha)}{x_A + x_B} \cdot p. 
    \end{align*}
    \item under the \textit{under-subscribed regime} if and only if:
    \begin{align*}
     &x_A =  x_A \cdot p + \frac{(1-x_A)(2\alpha - x_A-x_B)}{(2-x_A-x_B)}\cdot p + \frac{2(1-\alpha)(1-x_A)}{(2-x_A-x_B)}\cdot\epsilon,\\
     &x_B =  x_B \cdot p + \frac{(1-x_B)(2\alpha - x_A-x_B)}{(2-x_A-x_B)}\cdot p.
    \end{align*}
\end{itemize}
\noindent
The following result characterizes the fixed points of the system as a function of the affinity advantage parameter $\epsilon > 0$.

\begin{figure*}[ht]
  \centering
    \includegraphics[width=0.75\textwidth]{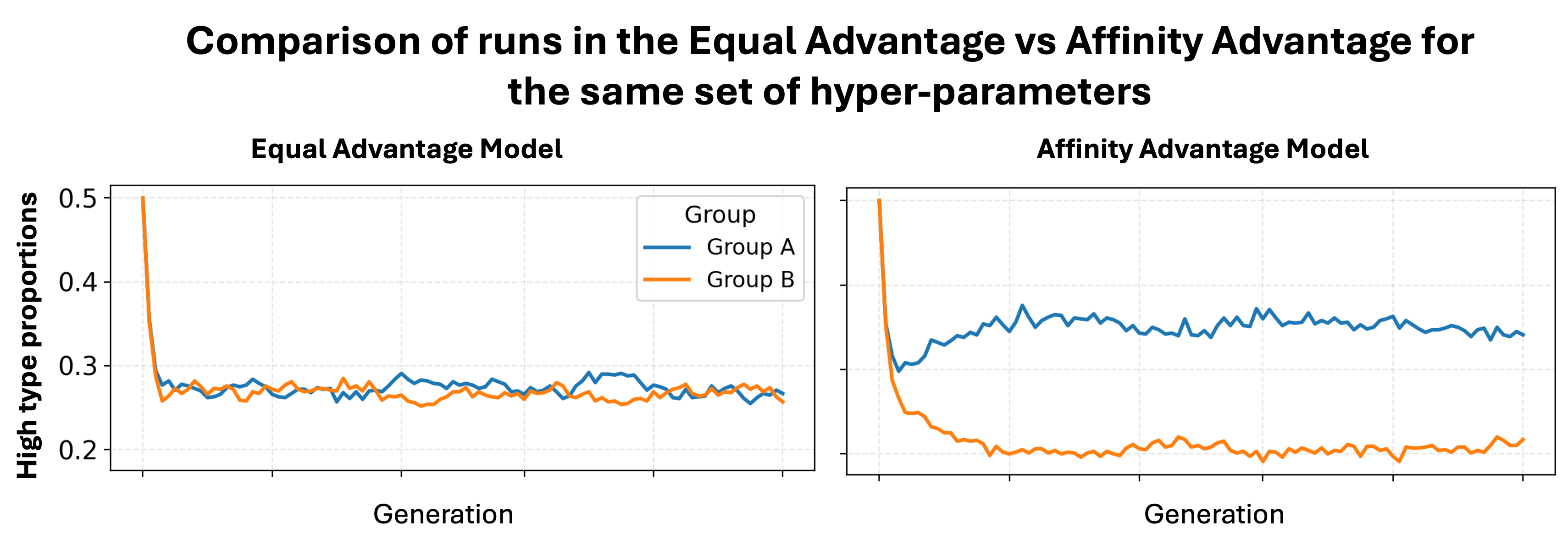}
    \caption{Sample trajectories showing the impact of adding a small affinity advantage ($\epsilon=0.03$): in the EA model (left), parity is restored, while in the AA model (right), the same advantage produces persistent separation under identical starting conditions.}
    \label{fig:ea_aa_comparison}
\end{figure*}

\begin{thm}\label{thm:affinity}
Let $\alpha < \frac{1}{2}$. Without loss of generality, we assume that group $A$ has the affinity advantage $\epsilon > 0$. There exists a threshold value of $\epsilon$ given by
    $\tilde{\epsilon} = \frac{2\alpha(1-p)}{(1-2\alpha)}$
such that:
\begin{enumerate}[label=\alph*)]
    \item for all $\epsilon \geq \tilde{\epsilon}$, the system has a unique fixed point $(x_A, x_B)$. It occurs in the over-subscribed regime (i.e., $x_A+x_B \geq 2\alpha$). The fixed point can be reached from any starting point $(x,y) \in [0,1]^2$ with $x>y$. Further, at the fixed point, we have $x_A = 2\alpha(p-\epsilon) + \epsilon$ and $x_B = 0$; the system achieves strictly positive separation. 
    \item for all $\epsilon < \tilde{\epsilon}$, the system has a unique fixed point $(x_A, x_B)$. It occurs in the under-subscribed regime (i.e., $x_A+x_B < 2\alpha$) and can be reached from any starting point $(x,y) \in [0,1]^2$ with $x>y$. Further, at the fixed point, we have $x_A > 1-\min \left(\frac{2\alpha(1-p)}{\epsilon}, \frac{(1-\alpha p)^2}{2(1-\alpha)\epsilon} \right)$; the system achieves a separation strictly larger than
    \[
                  \max \left(0, 2(1-\alpha) - \min \left(\frac{4\alpha(1-p)}{\epsilon}, \frac{(1-\alpha p)^2}{(1-\alpha)\epsilon} \right) \right).
    \]
\end{enumerate}
\end{thm}

Detailed proof can be found in Appendix~\ref{app:theory}. While the theorem characterizes behavior across a range of $\epsilon$, for small $\epsilon$, the lower bound can be vacuous. However, the dynamics of the system can still be solved precisely (through numerical experiments). Together, the theorem and the numerical solutions have two key implications. Firstly, even a small advantage ($\epsilon > 0$) guarantees a permanent, non-zero separation between the groups in the long run (Figure~\ref{fig:thm_4}). More importantly, the magnitude of $\epsilon$ dictates the severity of the disparity; a large enough advantage to one group can drive the system to an extreme outcome where the trailing group is completely excluded from college. As $\epsilon$ decreases, the extent of the disparity at the fixed point diminishes, until we recover Theorem~\ref{thm:equal_adv} in the limit $\epsilon \to 0$.

\begin{figure*}[h]
    \centering
    \includegraphics[width=0.6\textwidth]{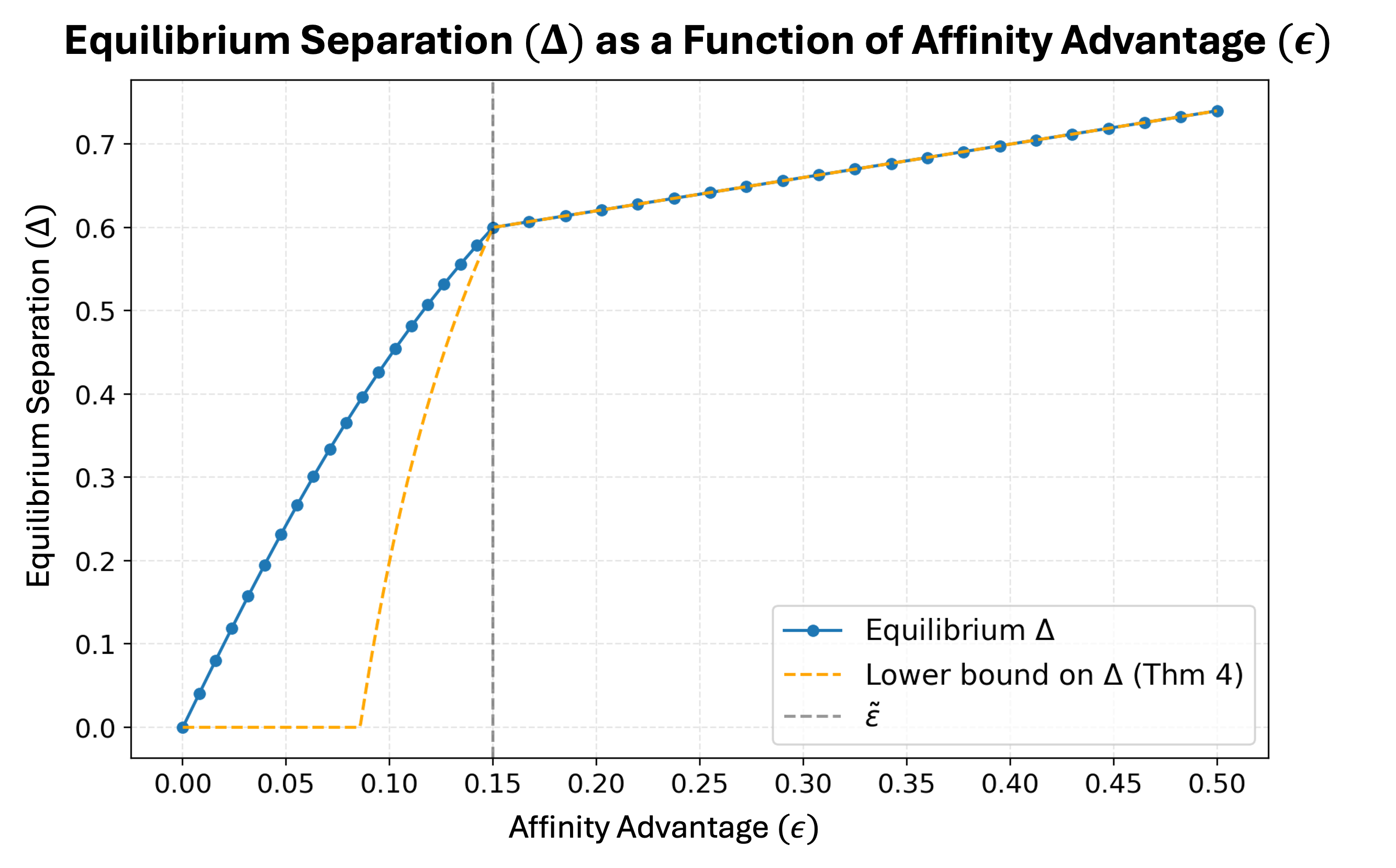}
    \caption{Equilibrium separation ($\Delta$, solid, blue, with markers) and the lower bound on separation (dashed, orange) from Theorem \ref{thm:affinity} as a function of the affinity advantage $\epsilon$. The vertical line marks the threshold value 
    $ \tilde{\epsilon} = \frac{2\alpha(1-p)}{(1-2\alpha)}$. Parameter combinations: $T=500, \alpha=0.30, p=0.90, q=0.0$. For very small $\epsilon$, the lower bound predicted by the theorem is vacuous, but it approximates the true equilibrium separation better for larger $\epsilon$ while matching it exactly for $\epsilon \geq \tilde \epsilon$. Note that separation is meaningful even for small $\epsilon$.}
    \label{fig:thm_4}
\end{figure*}

\section{A Richer Simulation Model}\label{sec:full_model}

To examine whether our theoretical insights carry over to richer environments, we now turn to a simulation that preserves the \emph{same primitives}—meritocratic selection to a scarce resource and inter-generational transmission of advantage—but swaps binary types for continuous abilities and allows stochastic variation. First, this allows us to check the \emph{robustness} of our simpler, theoretical model. Second, we note that the mechanisms that are expected to lead to separation map one-for-one to our theoretical model. Here, we study whether our qualitative predictions (no persistent separation without affinity unless populations are small, and persistent separation with affinity) remain true in this setting.

\begin{figure*}[!ht]
  \centering
    \includegraphics[width=0.6\textwidth]{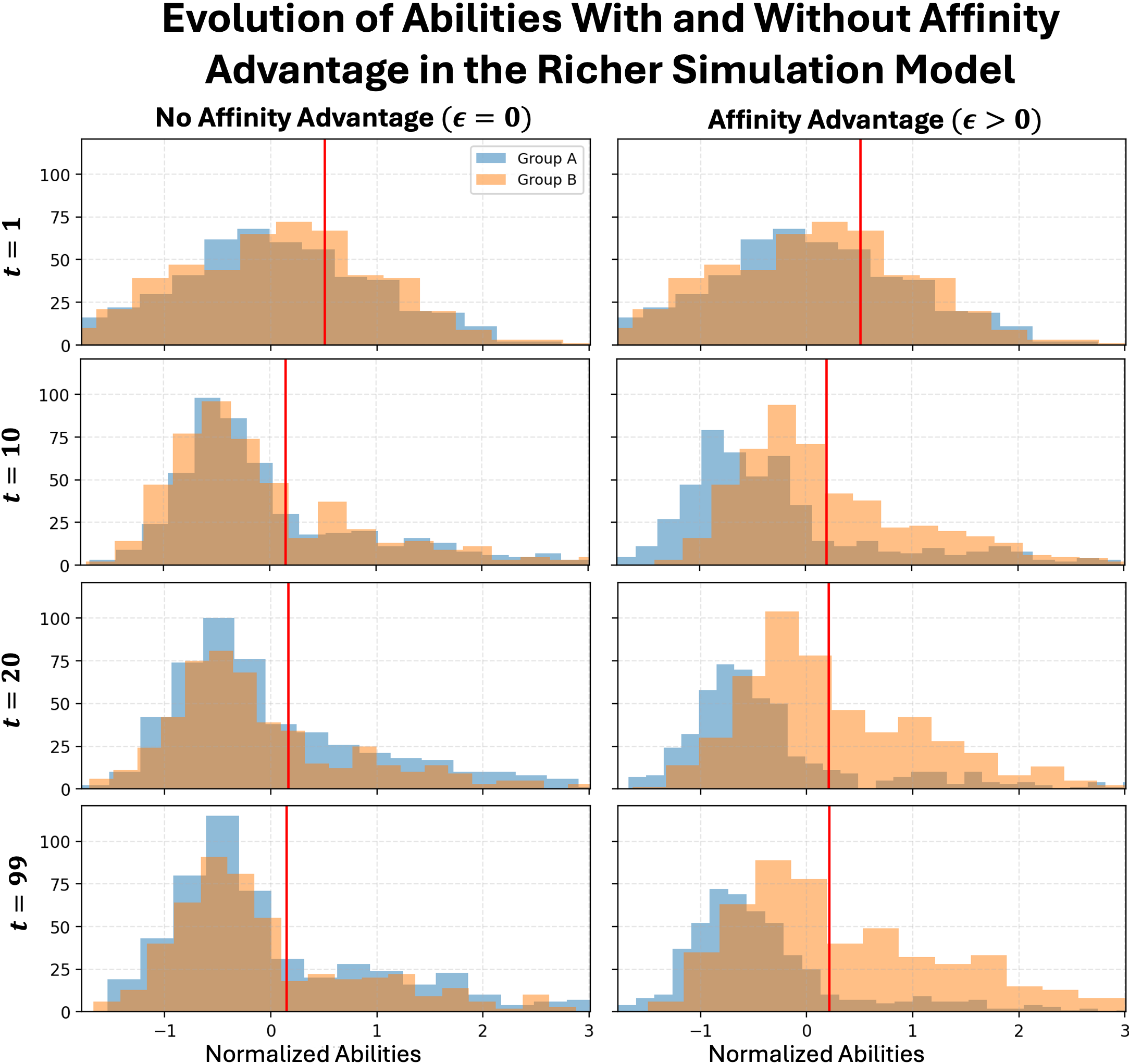}
    \caption{Depiction of how affinity advantage ($\mu_\epsilon$) affects ability distributions and admission outcomes over time. We compare two identical initial groups (group A and group B) under no advantage (left, $\mu_\epsilon = \sigma_\epsilon= 0$) and with advantage (right, $\mu_\epsilon = 0.20, \sigma_\epsilon = 0.10$) at time steps $t={1, 10, 20, 99}$. The vertical red line indicates the admission threshold. On the right, the ability distribution for one group progressively stochastically dominates the other.}
    \label{fig:trajectories_full}
\end{figure*}

\paragraph{Model}
We now model ability as continuous: for individual $j$, $a_j(t)$ denotes ability at time $t$, and $a_j(0)\sim\mathcal{N}(\mu_i,\sigma_i^2)$ for group $i\in{A,B}$. At each generation, a fraction $\alpha$ of the population is admitted to college by selecting the individuals with the highest \textit{standardized abilities}. Admitted students receive an ability boost $ I_j(t)\sim\mathrm{Gamma} \left(\tfrac{\mu_I^2}{\sigma_I^2},\tfrac{\sigma_I^2}{\mu_I}\right)$---this is analogous to the success probability \textbf{$p$} in the theoretical model. We choose a Gamma distribution, one of the two standard choices, along with the lognormal, for modeling non-negative random variables, because, in addition to modeling the assumed non-negative effect of college on ability, it also (i) is right-skewed with a tunable heavy tail, modeling situations where a small fraction of the population may see exceptional gains; and (ii) has a natural ``sum-of-increments'' interpretation (a Gamma can be viewed as the sum of many random variables---such as the combination of repeated contributions throughout one's college experience). 
Non-admits receive no boost $(I_j(t)=0)$. The next generation's ability is determined by the parent's final ability, a fraction of the college boost transmitted across generations (controlled by $\theta$, analogous to \textbf{$q$}), and a noise term $(\lambda_j(t) \sim \mathcal{N}(0, \sigma^2_{\lambda}))$ which models the imperfect transmission of ability from parent to child, representing the natural variation between generations:
\begin{equation*}
    a_j(t+1) = a_j(t) + \theta \cdot I_j(t) + \lambda_j(t).
\end{equation*}

\paragraph{Conferring Affinity Advantage}
To mirror our Affinity Advantage (AA) model from earlier, we introduce a feedback loop: if one group has more members admitted to college in generation $t$, then in the next generation every non-college admit from that leading group receives an additional stochastic ability boost, $\epsilon_j(t+1)$, in the next generation, drawn from $\mathcal{N}(\mu_\epsilon, \sigma_\epsilon^2)$. College-admits (who already receive $I_j$) and individuals in the trailing group do not receive this.

\begin{figure*}[!ht]
  \centering
    \includegraphics[width=0.7\textwidth]{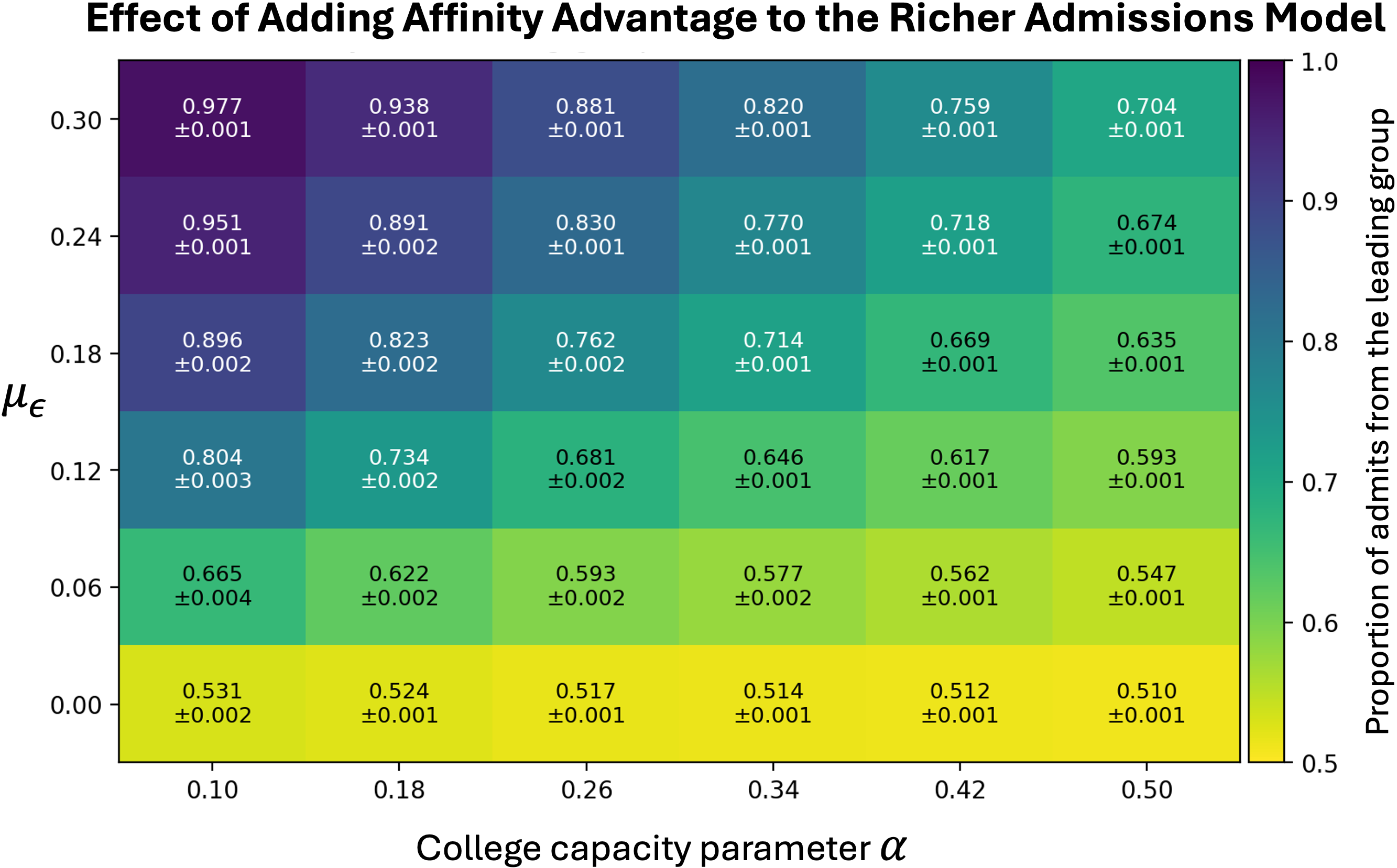}
    \caption{Heatmap showing the share of the leading group (with higher mean ability) out of all college admits as a function of the capacity parameter $\alpha$ (x-axis) and mean affinity advantage $\mu_\epsilon$ (y-axis). Cells show mean $\pm$ standard error calculated over $100$ independent runs. Experiment initialization: two identical groups $A$ and $B$ ($\mu=0$, $\sigma=1$, $n_A=n_B=500$); college admits the top $\alpha \cdot (n_A + n_B)$ students in the population and boosts each admitted individual's ability ($\mu_I=1.0,\sigma_I=1.0$). Affinity advantage ($\epsilon \sim \mathcal{N}(\mu_\epsilon,\sigma_\epsilon^2)$ with $\sigma_\epsilon=\mu_\epsilon/2$) to non college-goers from leading group. Abilities of individuals in each generation updated as defined in Section \ref{sec:full_model}. System simulated for $T = 500$ generations. Higher capacity improves long-term outcomes with respect to fairness with more equitable shares of admitted students from each group.}
    \label{fig:admit_proportions_full}
\end{figure*}

For our experiments, we set both groups to be initially identical ($\mu_A = \mu_B = 0,~\sigma_A = \sigma_B = 1$) with base parameters $\theta=~\mu_I = \sigma_I = 1,$ and $\sigma_{\lambda} = 0.5$. To demonstrate that the insights from our EA and AA models are not just artifacts of their simplifying assumptions, we test their core predictions in this richer simulation model. 
We observe even in this model that any small advantage for one group is amplified over time, creating a reinforcing cycle that leads to one group dominating (Figure \ref{fig:trajectories_full} shows an example).

We quantitatively assess long-run behavior by averaging admission shares over the last 30 generations across 100 runs on a grid of $\alpha$ and $\mu_\epsilon$ values (Fig.~\ref{fig:admit_proportions_full}). At $\mu_\epsilon=0$ we observe near parity (shares $\approx 0.5$). Already at $\mu_\epsilon=0.06$, the advantaged group captures $\approx0.6$ of admits on average, evidencing substantial separation; for $\mu_\epsilon \ge 0.24$ we frequently see $>0.7$, indicating greater-dominance. Increasing capacity $\alpha$ attenuates these effects but does not mitigate them completely. These findings reinforce that the main insights of our simple model are robust even under more realistic model extensions and inter-generational dynamics. Sensitivity analysis confirms that long-term group separation is a robust phenomenon that persists across a wide range of college efficacy levels $(\mu_I)$ and generational noise magnitudes $(\sigma_\lambda)$ (See Appendix ~\ref{app:sensitivity} for details).

\section{Discussion and Conclusion}\label{sec:conclusion}
We explore how disparities across groups can emerge and persist in systems with limited resources and inter-generational feedback loops. Importantly, disparities can arise in the long-term \emph{even when there are no disparities across the initial population} and when \emph{fair and meritocratic} selection rules are used. This highlights that considering fairness in a static sense can be insufficient to understand long-term disparities, \emph{even in best-case initial scenarios}. Even establishing parity across groups some time in the long run is also insufficient, as inter-generational dynamics and stochasticity can cause the system to drift away again, re-creating disparities, thus necessitating more pro-active approaches to fairness. Finally, with affinity advantages (for example, positive network effects from being part of an advantaged group), the findings are more concerning: group disparities quickly become persistent, with the potential to exclude one group from access to the resource in the worst case. 
We have briefly highlighted that increasing college capacity and access to education has the potential to alleviate these effects by preventing the formation of an elite class where network effects are strong and can amplify group disparities. However, much remains to be done in the space of designing policy interventions to mitigate such effects.

\paragraph{Limitations of Modeling Choices.}
Our transition matrix is intentionally simple to keep the feedback mechanism transparent: in the EA baseline, rejected low types do not upgrade, rejected high types retain their types with probability $q$, admissions are meritocratic, and capacity $\alpha$ is fixed—even though in practice signals are noisy and capacity can respond to demand. One could easily imagine richer models along a range of dimensions, for example allowing baseline upward mobility for rejected low types, imperfect persistence for rejected high types, mean-reverting mechanisms, and other more nuanced affinity mechanisms. However, as our theoretical model captures in its essence, and our richer simulations provide further evidence for, even under these relaxations, the same qualitative picture should emerge because \emph{three primitives remain consistent}: scarce, rank-based selection magnifies small leads; inter-generational transmission preserves realized advantages; and any positive group-level feedback pushes the leading group further ahead.

\paragraph{Impact.} The implications of our findings are multi-fold. First, they reveal the limitations of static, one-time fairness evaluations and interventions, underlining the need to account for temporal dynamics and stochasticity while making fairness assessments. Second, they identify new failure points (randomness and network effects) that could be targeted for interventions to promote long-term fairness. 
Finally, we believe that our findings also have implications in the context of algorithmic decision-making. The “randomness” in our model is abstract and covers a multitude of real mechanisms by which algorithmic decision pipelines can introduce stochasticity (measurement noise, model error, cutoff effect for borderline cases). For example, in algorithmic monoculture, similar algorithms deployed across multiple domains can learn to specialize on specific populations, creating self-reinforcing feedback loops where those groups get favored/disfavored disparately. This is analogous to the concept of affinity advantage (that we introduce in our model) and is expected to produce similar outcomes, only this is driven by monoculture and not by network effects. This demonstrates that insights from our model can be generalized to much broader settings.

\section*{Acknowledgements}
The authors are grateful for support from the US National Science Foundation (NSF) under grants  IIS-2504990, IIS-2336236 and IIS-2533162. Any opinions and findings expressed in this material are those of the authors and do not reflect the views of their funding agencies.

\bibliographystyle{plainnat} 
\bibliography{mybib}

\newpage
\appendix
\newpage
\onecolumn

\section{Proofs of Theoretical Results}\label{app:theory}

\subsection{Proof of Theorem~\ref{thm:allocation}}
We will first analyze the over-subscribed regime followed by the under-subscribed regime. 
\paragraph{Over-subscribed regime.} When $X_A(t) + X_B(t) \geq 2\alpha$, note that a meritocratic selection rule must allocate only to high-type candidates. We now impose the fairness criteria which requires us to equalize type-wise selection rates across groups. The selection rate among low-type candidates is already $0$ across both groups (and thus trivially equalized). 

Now, let $\mathcal{A}_i(t+1)$ indicate the number of seats allocated to high-type candidates at time $t+1$ for group $i$, $i \in \{A, B\}$. Firstly, it must be that $\sum_{i \in \{A, B\}}\mathcal{A}_i(t+1) = C = 2N\alpha$ (all seats have to be allocated for efficiency). In order for the high-type selection rate to be equalized across groups, we need: 
\[
       \frac{\mathcal{A}_A(t+1)}{NX_A(t)} = \frac{\mathcal{A}_B(t+1)}{NX_B(t)}.
\]
By the addendo property, the high-type selection rate in a fair allocation rule must equal:
\[
           \frac{\mathcal{A}_A(t+1) + \mathcal{A}_B(t+1)}{N(X_A(t) + X_B(t))} = \frac{2N\alpha}{NX(t)} = \frac{2\alpha}{X(t)},
\]
which further implies that: 
\[
         \mathcal{A}_i(t+1) = \frac{X_i(t)}{X(t)}(2N\alpha), \quad \forall~i \in \{A, B\}.
\]
This is exactly our allocation rule in the over-subscribed regime. 

\paragraph{Under-subscribed regime.} When $X_A(t) + X_B(t) < 2\alpha$, in order for the allocation to be meritocratic and efficient, all high-types across both groups have to be admitted first and remaining seats need to be allocated to low-type candidates from both groups. Note that the selection rate for high-types across both groups is equal to $1$ and is thus trivially equalized. 

Finally, we need to enforce the fairness criteria on the selection rates of low-types during the allocation of residual seats $(C- NX(t))$. Let $\widetilde{\mathcal{A}}_i(t+1)$ indicate the seats allocated to low-type candidates from group $i$ at time $t+1$ in the under-subscribed regime. Therefore, we need: 
\[
       \frac{ \widetilde{\mathcal{A}}_A(t+1) }{N(1-X_A(t))} =  \frac{ \widetilde{\mathcal{A}}_B(t+1) }{N(1-X_B(t))}.
\]
Using the addendo property again, the selection rate among low-types in a fair allocation must equal: 
\[
         \frac{ \widetilde{\mathcal{A}}_A(t+1) + \widetilde{\mathcal{A}}_B(t+1)  }{N(1-X_A(t)) + N(1-X_B(t))} = \frac{(C - NX(t))}{N(2-X(t))},
\]
which implies that: 
\[
       \widetilde{\mathcal{A}}_i(t+1) = \frac{1-X_i(t)}{2-X(t)}\cdot (C-NX(t)) \quad \forall~i \in \{A, B\}. 
\]
This again is our exact allocation rule in the under-subscribed regime. This concludes the proof. 

\subsection{Proof of Theorem~\ref{thm:equal_adv}}
The proof of Theorem~\ref{thm:equal_adv} will be completed in the following steps: 
\begin{enumerate}
    \item First, we will show that under the Equal Advantage (EA) model, the system has 
    exactly one fixed point $(x_A, x_B)$ in the under-subscribed regime with $x_A = x_B =\alpha p$. 
    \item Then we show that for any starting point $(X_A(0), X_B(0)) \in [0,1]^2$, the system will always reach the above fixed point. This also implies that the above fixed point must be the unique fixed point of the system.  
\end{enumerate}

\paragraph{There exists one fixed point in the under-subscribed regime under the EA model.} By the definition of a fixed point and the system dynamics in the under-subscribed case, we have: 
\begin{align*}
    x_A &= x_A \cdot p + (2\alpha - x_A - x_B)\cdot \frac{(1-x_A)}{(2-x_A - x_B)}\cdot p\\
    x_B &= x_B \cdot p + (2\alpha - x_A - x_B)\cdot \frac{(1-x_B)}{(2-x_A - x_B)}\cdot p.
\end{align*}
Adding both equations together, we obtain $x_A + x_B = 2\alpha p$. Plugging it back into the individual equations, we obtain $x_A = x_B = \alpha p$ uniquely. Clearly, this must the only fixed point of the system in the under-subscribed regime. \\ 

\noindent
Putting the two results together, we conclude that the system has a unique fixed $(x_A, x_B)$ which occurs in the under-subscribed regime and is given by $x_A = x_B = \alpha p$. 

\paragraph{Given any starting point $(X_A(0), X_B(0))$ and $0 < q < p < 1$, the system always reaches the fixed point.} For ease of exposition, let us introduce some notation. Let $x_A(0) = X_A(0)$, $x_B(0) = X_B(0)$. Further, for any $t \geq 1$, we define: 
\[
     \left( x_A(t), x_B(t) \right) = \mathcal{T}\left(x_A(t-1), x_B(t-1) \right), 
\]
where $\mathcal{T}$ captures the transition dynamics of the EA model. We now make the following observations: 
\begin{itemize}
    \item If $x_A(t) + x_B(t) \geq 2\alpha$, $x_A(t+1) + x_B(t+1) < x_A(t) + x_B(t)$. This follows trivially because: 
    \[
      x_A(t+1) + x_B(t+1) = 2\alpha (p-q) + q \cdot \left(x_A(t) + x_B(t)\right) < x_A(t) + x_B(t)
    \]
    since $p < 1$ and $x_A(t) + x_B(t) \geq 2\alpha$ to begin with. The amount of decrease $\delta$ in each time step is given by: 
    \begin{align*}
        \delta &= \left( x_A(t) + x_B(t) \right) - \left( x_A(t+1) + x_B(t+1) \right) \\
        &= (1-q)\cdot \left(x_A(t) + x_B(t) \right) - 2\alpha (p-q)\\
        &\geq 2\alpha(1-q) - 2\alpha(p-q) = 2\alpha(1-p). 
    \end{align*}
    Since $\delta$ is finite (and not vanishingly small), there exists a finite time $\tau > t$ such that $x_A(\tau) + x_B(\tau) < 2\alpha$. 
    \item If $x_A(t) + x_B(t) < 2\alpha$, 
    \[
        x_A(t+1) + x_B(t+1) = 2 \alpha p.
    \] 
    \item Combining the above observations, we conclude that for any starting point, $x_A(t) + x_B(t)$ always converges to $2\alpha p$ in finite time steps.
    \item Once $x_A(t) + x_B(t)$ reaches $2\alpha p$, we have: 
    \[
        x_A(t+1) = x_A(t) + \frac{\alpha p(1-p)}{(1-\alpha p)} - \frac{(1-p)}{(1-\alpha p)}\cdot x_A(t),
    \]
    and similarly for $x_B(t+1)$. Observe that $x_A(t) > \alpha p$ implies $\alpha p < x_A(t+1) < x_A(t)$, but if $x_A(t) < \alpha p$, $\alpha p > x_A(t+1) > x_A(t)$. Also, if $x_A(t) = \alpha p$, $x_A(t+1) = \alpha p$. Clearly, $x_A(t)$ must eventually converge to $\alpha p$ (by the monotone convergence theorem). We can argue similarly for $x_B(t)$. 
\end{itemize}  
Thus, we have shown that for any starting point, the dynamics converge to the unique fixed point $(\alpha p, \alpha p)$. This concludes the proof of the theorem.  

\subsection{Stochastic Setting: Preliminaries}
$\{ \left(X_A(t), X_B(t) \right) \}$ is a stochastic process. According to the dynamics of the EA model, we can actually derive intuition about the distributions of $X_A(t+1)$ and $X_B(t+1)$, given the system state at time $t$. \\

\noindent
\textit{Over-subscribed regime.} We know that in the over-subscribed regime (i.e., $X_A(t) + X_B(t) \geq 2\alpha$), we have: 
\[
       X_A(t+1) = \frac{1}{N}\left(Y_A(t+1) + Z_A(t+1)\right), \quad Y_A(t+1) \sim Bin\left(\frac{(2N\alpha)X_A(t)}{X(t)}, p\right) 
\]
\[
       Z_A(t+1) \sim Bin\left(NX_A(t) - \frac{(2N\alpha)X_A(t)}{X(t)}, q\right).
\]
\[
    X_B(t+1) = \frac{1}{N}\left(Y_B(t+1) + Z_B(t+1)\right), \quad Y_B(t+1) \sim Bin\left(\frac{(2N\alpha)X_B(t)}{X(t)}, p\right)
\]
\[
    Z_B(t+1) \sim Bin\left(NX_B(t) - \frac{(2N\alpha)X_B(t)}{X(t)}, q\right). 
\]
For $N$ reasonable large, we can take the normal approximation of the binomial distribution (in fact, the binomial distribution approaches normal quickly, so the approximation works very well even without taking $N$ to $\infty$). Further, since $Y_i(t+1)$ and $Z_i(t+1)$ are independent, we have the following distribution for $X_i(t+1) ~\forall~i \in \{A, B\}$: 
\begin{align*}
    X_i(t+1) \sim \mathcal{N}\left(qX_i(t) + \frac{2\alpha(p-q)X_i(t)}{X(t)}, \frac{1}{N}\left( \frac{2\alpha X_i(t)}{X(t)}\left[p(1-p) - q(1-q) \right] + X_i(t)q(1-q)  \right)    \right). 
\end{align*}
We define the separation between groups $A$ and $B$ at time $t$ as $\Delta(t) = X_A(t+1)-X_B(t+1)$. Since $X_A(t+1)$ and $X_B(t+1)$ are also independent, $\Delta(t+1)$ is also normally distributed as follows: 
\begin{align}\label{eq:1}
    \Delta(t+1) \sim \mathcal{N}\left( \left(\frac{2\alpha(p-q)}{X(t)} + q\right)\Delta(t), \frac{1}{N}\left( 2\alpha p(1-p) + q(1-q)\left(X(t)-2\alpha \right) \right)  \right).
\end{align}

\noindent
\textit{Under-subscribed regime.} We can do an identical analysis for the under-subscribed regime where $X_A(t) + X_B(t) < 2\alpha$. There, we can show: 
\begin{align*}
    X_A(t+1) \sim \mathcal{N}\left( pX_A(t) + \frac{2\alpha - X(t)}{2-X(t)}(1-X_A(t))p, \frac{p(1-p)}{N}\left(X_A(t) + \frac{2\alpha-X(t)}{2-X(t)}(1-X_A(t)) \right)  \right), \\
    X_B(t+1) \sim \mathcal{N}\left( pX_B(t) + \frac{2\alpha - X(t)}{2-X(t)}(1-X_B(t))p, \frac{p(1-p)}{N}\left(X_B(t) + \frac{2\alpha-X(t)}{2-X(t)}(1-X_B(t)) \right)  \right).
\end{align*}
Therefore, 
\begin{align}\label{eq:2}
    \Delta(t+1) \sim \mathcal{N}\left( \frac{2(1-\alpha)p}{2-X(t)}\Delta(t), \frac{2\alpha p(1-p)}{N} \right).
\end{align}
\textbf{Notes:} Eqs~\eqref{eq:1} and \eqref{eq:2} are the \textbf{conditional} distributions of $\Delta(t+1)$, conditional on $X_A(t)$ and $X_B(t)$. Also, our proofs in Theorems~\ref{thm:time_parity} and \ref{thm:separation} work under the $N$ sufficiently large assumption where the normal approximation works!

\subsection{Proof of Theorem~\ref{thm:time_parity}.}
Given $X_A(0)$ and $X_B(0)$, we want to compute the time required to reach $\eta$-parity. For any $t$, we use the short-hand $X(t) = X_A(t) + X_B(t)$ and $\Delta(t) = X_A(t) - X_B(t)$.  

\paragraph{Under-subscribed regime.} Suppose, $X(t) < 2\alpha$ and $|\Delta(t)| > \eta$. Now, 
recall from the previous segment that the gap $\Delta(t+1)$ can be shown to be normally distributed, conditional on $X_A(t)$ and $X_B(t)$, when $N$ is large. Reusing the result in the under-subscribed regime, we have: 
\[
           \Delta(t+1) \sim \mathcal{N}\left(\frac{2(1-\alpha)p}{2-X(t)}\Delta(t), \frac{2\alpha p(1-p)}{N}\right).
\]
Since, $X(t) < 2\alpha$, we have: 
\[
        \mathbb{E}\left[ \Delta(t+1)~|~X_A(t), X_B(t) \right] = \frac{2(1-\alpha)p}{2-X(t)}\Delta(t) < p\cdot \Delta(t). 
\]
Note that by Jensen's inequality, for any random variable $Z$, we have: $\mathbb{E}^2[ |Z| ] \leq \mathbb{E}[Z^2] = \mathbb{E}^2[Z] + Var(Z)$. 
Therefore, we have: 
\begin{align*}
    \mathbb{E}^2\left[~|\Delta(t+1)| ~|~X_A(t), X_B(t) \right] &\leq \mathbb{E}^2\left[ \Delta(t+1)~|~X_A(t), X_B(t) \right] + Var(\Delta(t+1)~|~X_A(t), X_B(t)) \\
    &< p^2 \cdot \Delta(t)^2 + \frac{2\alpha p(1-p)}{N} \\
    &= p^2 \cdot \Delta(t)^2 + 2\cdot p |\Delta(t)| \cdot \frac{\alpha (1-p)}{N |\Delta(t)|}\\
    &< \left(p |\Delta(t)| + \frac{\alpha (1-p)}{N |\Delta(t)|} \right)^2.
\end{align*}
This implies that: 
\begin{align*}
       \mathbb{E}\left[~|\Delta(t+1)| ~|~X_A(t), X_B(t) \right] < p |\Delta(t)| + \frac{\alpha (1-p)}{N |\Delta(t)|} < p|\Delta(t)| + \frac{\alpha(1-p)}{N \eta}.
\end{align*}
Therefore, we have: 
\[
       \mathbb{E}\left[~|\Delta(t+1) |~ \right] < p \cdot \mathbb{E}[~|\Delta(t)|~] + \frac{\alpha(1-p)}{N \eta}. 
\]

\paragraph{Over-subscribed regime.} Now, suppose that $X(t) \geq 2\alpha$ and $|\Delta(t)| > \eta$. Recall that for the over-subscribed regime, the conditional distribution of the gap $\Delta(t+1)$ is normal when $N$ is large. 
\[
    \Delta(t+1) \sim \mathcal{N}\left( \left(\frac{2\alpha(p-q)}{X(t)} + q\right)\Delta(t), \frac{1}{N}\left( 2\alpha p(1-p) + q(1-q)(X(t)-2\alpha)\right) \right).
\]
Since $X(t) \geq 2\alpha$, we again have: 
\[
     \mathbb{E}[\Delta(t+1)~|~X_A(t), X_B(t)] = \left(\frac{2\alpha(p-q)}{X(t)} + q\right)\Delta(t) \leq p \cdot \Delta(t).  
\]
\[
       Var(\Delta(t+1)~|~X_A(t), X_B(t)) = \frac{1}{N}\left( 2\alpha p(1-p) + q(1-q)(X(t)-2\alpha)\right) \leq \frac{2}{N}\left( \alpha p(1-p) + (1-\alpha)q(1-q) \right).
\]
Using an identical analysis as the under-subscribed regime, we can obtain an upper bound on $\mathbb{E}[~|\Delta(t+1)|~]$ as follows: 
\begin{align*}
    \mathbb{E}^2 \left[~|\Delta(t+1)|~|~X_A(t), X_B(t) \right] &\leq p^2 \cdot \Delta(t)^2 + \frac{2}{N}\left( \alpha p(1-p) + (1-\alpha)q(1-q) \right) \\
    &< \left( p\cdot |\Delta(t)| +  \frac{\alpha(1-p)}{N \eta} + \frac{(1-\alpha)q(1-q)}{Np \eta} \right)^2,
\end{align*}
which eventually implies that: 
\[
    \mathbb{E}\left[~|\Delta(t+1)| ~\right] < p \cdot \mathbb{E}[~|\Delta(t)|~] + \frac{\alpha(1-p)}{N \eta} + \frac{(1-\alpha)q(1-q)}{Np \eta}. 
\]
Therefore, across both regimes, we obtain an upper bound on $\mathbb{E}\left[~|\Delta(t+1)| ~\right]$ in terms of $\mathbb{E}\left[~|\Delta(t)| ~\right]$ which we can then use inductively to obtain an upper bound in terms of $\Delta(0)$. We have the following recursion relation: 
\begin{align*}
    \mathbb{E}\left[~|\Delta(t+1)|~ \right] &< p \cdot \mathbb{E}\left[~|\Delta(t)|~\right] + \frac{f(\alpha, p, q)}{N \eta} \\
    &< p^{t+1}\cdot |\Delta(0)| + \frac{f(\alpha, p, q)}{N \eta}\cdot \frac{(1-p^{t+1})}{(1-p)}\\
    &= p^{t+1}\cdot \left( |\Delta(0)| - \frac{f(\alpha, p, q)}{N \eta (1-p)} \right) + \frac{f(\alpha, p, q)}{N\eta (1-p)}. 
\end{align*}
We want to find the number of time steps required for the system to reach $\eta$-parity for the first time with probability $\geq 1-\omega$. By Markov's inequality, we have: 
\begin{align*}
    &\mathbb{P}\left[~|\Delta(t+1)| > \eta \right] \leq \frac{\mathbb{E}[~|\Delta(t+1)|~]}{\eta} < p^{t+1}\cdot\left(\frac{|\Delta(0)|}{\eta} - \frac{f(\alpha,p,q)}{N\eta^2 (1-p)} \right) + \frac{f(\alpha,p,q)}{N\eta^2 (1-p)} < \omega \\
    \implies & (t+1)\log(p) < \log\left( \omega - \frac{f(\alpha,p,q)}{N\eta^2 (1-p)} \right) - \log\left(\frac{|\Delta(0)|}{\eta} - \frac{f(\alpha,p,q)}{N\eta^2 (1-p)}\right) \\
    \implies & t > \left \lceil \frac{1}{\log(p)} \left[ \log\left( \omega - \frac{f(\alpha,p,q)}{N\eta^2 (1-p)} \right) - \log\left(\frac{|\Delta(0)|}{\eta} - \frac{f(\alpha,p,q)}{N\eta^2 (1-p)}\right)  \right]  \right \rceil - 1.
\end{align*}
Thus, for all $t$ at least as large as the above value, the system will reach $\eta$-parity by $t$ time steps with probability $\geq 1- \omega$.

\subsection{Separation: A Theoretical Treatment}

\begin{thm}\label{thm:separation}
Let $T_{\delta}$ be the number of time steps needed by the system to reach $\delta$-separation with probability $\geq 1-\omega$. Then, 
\[
       T_{\delta} \leq \left\lceil \frac{\log(1/\omega)}{p_s} \right\rceil,
\]
where $p_s$ is given by: 
\begin{align*}
    p_s = 2 \cdot \min \left\{ \Phi^c \left( \frac{ \sqrt{N}(1+\delta - p(1-\alpha)) }{\sqrt{2\alpha p(1-p)} } \right),
    \Phi^c \left( \frac{\sqrt{N}(1+\delta-q-\alpha(p-q))}{ \sqrt{2\alpha p(1-p)} }\right) \right\}
\end{align*}
and $\Phi^c(\cdot)$ indicates the complementary standard normal CDF. 
\end{thm}
\noindent
\textit{Proof.} Before we proceed with the proof of the theorem, let us
introduce a \textit{stronger} notion of $\delta$-separation as follows: 
\begin{defn}\label{defn:one_step_sep}
The system is said to have achieved \textit{one-step} $\delta$-separation at $\tau$ if:
\[
        \tau:= \arg \min \left\{t:~ \left| \Delta (t) - \Delta(t-1)\right| \geq \delta \right\}.
\]
\end{defn}
It must be clear that the probability of the one-step separation event provides a conservative lower bound on the probability of the original $\delta$-separation event. It also eliminates the challenge of tracking how the system evolves and transitions across regimes as we can now do a one-step computation at time $t+1$, assuming knowledge of the system state at time $t$. \\

Suppose, the system is fully characterized at time $t$ by the tuple $(X_A(t), X_B(t))$ where $X_A(t) - X_B(t) = \Delta(t)$. We want to find a lower bound on the probability that one-step $\delta$-separation is achieved at time $t+1$. We can re-use the conditional distributions for $\Delta(t+1)$ derived earlier for both over-subscribed and under-subscribed regimes:

\paragraph{Over-subscribed regime.}
\begin{align*}
    \Delta(t+1) \sim \mathcal{N}\left( \left(\frac{2\alpha(p-q)}{X(t)} + q \right)\Delta(t), \frac{1}{N}\left( 2\alpha p(1-p) + q(1-q)\left(X(t)-2\alpha \right) \right)  \right).
\end{align*}
Therefore, 
\begin{align*}
    \mathbb{P}\left[ \Delta(t+1)-\Delta(t) \geq \delta \right] &= \Phi^c \left( \frac{\delta - \left(q + \frac{2\alpha(p-q)}{X(t)}-1 \right)\Delta (t)} { \frac{1}{\sqrt{N}}\sqrt{2\alpha p(1-p) + (X(t)-2\alpha)q(1-q)} } \right)\\
    &> \Phi^c \left( \frac{\sqrt{N}(\delta + 1-q-\frac{2\alpha(p-q)}{X(t)} )}{ \sqrt{2\alpha p(1-p)} }  \right) \quad \text{( $\because q + 2\alpha(p-q)/X(t) - 1 < 0,~ |\Delta(t)| \leq 1$  )} \\
    &> \Phi^c \left( \frac{\sqrt{N}(\delta + 1-q-\alpha(p-q))}{ \sqrt{2\alpha p(1-p)} }  \right) \quad \text{ ($\because X(t) \leq 2$) } \\
    &\geq \Phi^c \left( \frac{\sqrt{N}(1+\delta-q-\alpha(p-q))}{ \sqrt{2\alpha p(1-p)} }\right).
\end{align*}
This gives us the following two-sided bound: 
\[
     \mathbb{P}\left[ \left|\Delta(t+1)-\Delta(t) \right| \geq \delta \right] > 2 \cdot \Phi^c \left( \frac{\sqrt{N}(1+\delta-q-\alpha(p-q))}{ \sqrt{2\alpha p(1-p)} }\right).
\]

\paragraph{Under-subscribed regime.} Similarly, in the under-subscribed regime,  
\begin{align*}
    \Delta(t+1) \sim \mathcal{N}\left( \frac{2p(1-\alpha)}{2-X(t)}\Delta(t), \frac{2\alpha p(1-p)}{N} \right).
\end{align*}
Which implies that: 
\begin{align*}
    \mathbb{P}\left[ \Delta(t+1)-\Delta(t) \geq \delta \right] &=  \Phi^c \left(  \frac{\delta - \Delta(t)\left[\frac{2p(1-\alpha)}{(2-X(t))} - 1\right] }{\frac{1}{\sqrt{N}} \sqrt{2\alpha p(1-p)} } \right) \\
    &= \Phi^c \left( \frac{ \sqrt{N}(1+\delta - \frac{2p(1-\alpha)}{(2-X(t)))} }{ \sqrt{2\alpha p(1-p)} } \right) \\
    &> \Phi^c \left( \frac{ \sqrt{N}(1+\delta - p(1-\alpha)) }{\sqrt{2\alpha p(1-p)}} \right). 
\end{align*}
Combining both regimes, we get that: 
\begin{align*}
    \mathbb{P}\left[ \left|\Delta(t+1) - \Delta(t) \right| \geq \delta \right] > \underbrace{2 \cdot\min \left( \Phi^c \left( \frac{ \sqrt{N}(1+\delta - p(1-\alpha)) }{\sqrt{2\alpha p(1-p)}} \right), \Phi^c \left( \frac{\sqrt{N}(1+\delta-q-\alpha(p-q))}{ \sqrt{2\alpha p(1-p)} }\right) \right)}_{p_s}.
\end{align*}
Now, 
\begin{align*}
    \mathbb{P}\left[ \text{$\delta$-separation occurs over $T$ time steps} \right] &= 1-\mathbb{P}\left[ \text{no $\delta$-separation occurs for any time step in $[0, T]$} \right] \\
    &= 1 - \left(\text{no $\delta$-separation occurs in one time step} \right)^T \\
    &> 1- (1-p_s)^T \\
    &> 1 - \exp(-p_s T)
\end{align*}
If we want the above probability to be at least $1-\omega$ for some pre-chosen level of $\delta_0$, then: 
\begin{align*}
   1 - \exp(-p_s T) \geq 1- \delta_0 &\implies \exp(-p_s T) \leq \omega \\
   &\implies -p_s T \leq \log(\omega) \\
   &\implies T \geq \frac{\log(1/\omega)}{p_s}.
\end{align*}
Thus, we have shown that for all $T \geq \left\lceil \frac{\log(1/\omega)}{p_s} \right\rceil $, $\delta$-separation will be achieved sometime over $T$ time steps with probability at least $1-\delta_0$. This also implies that the number of time steps needed to achieve $\delta$-separation for the first time with high probability is given by: 
\[
           T_{\delta} = \left\lceil \frac{\log(1/\omega)}{p_s} \right\rceil. 
\]
This concludes the proof of the theorem.

\subsection{Proof of Theorem~\ref{thm:affinity}.}
We will complete the proof of Theorem~\ref{thm:affinity} in the following steps: 
\begin{enumerate}
    \item First, we will show that when $\epsilon \geq \tilde \epsilon$ and is in favour of group $A$, the system has a fixed point $(x_A, x_B)$ in the over-subscribed regime, given by $x_A = 2\alpha(p-\epsilon) + \epsilon$, $x_B = 0$ (Claim~\ref{clm:1}). 
    \item Next, we will show that when $\epsilon \geq \tilde \epsilon$, for any starting point where group $A$ is better off, the system will always reach the above fixed point. This should immediately imply that the above fixed point is the unique fixed point of the system when $\epsilon \geq \tilde \epsilon$ (Lemma~\ref{lem:1}).
    \item However, when $\epsilon < \tilde \epsilon$, we will have to follow a different approach because we cannot characterize the fixed point in closed form. So, we will first establish existence of a fixed point in the under-subscribed regime by Brouwer's fixed point theorem (Lemma~\ref{lem:brouwer}). Then, we will argue that this must be the only fixed point in the under-subscribed regime and characterize its properties (Lemma~\ref{lem:fp_under}). 
    \item Finally, we will show that for any starting point, the system always converges to the only fixed point of the under-subscribed regime when $\epsilon < \tilde \epsilon$, further implying this must be the unique fixed point of the system when $\epsilon < \tilde \epsilon$ (Lemma~\ref{lem:unique}).  
\end{enumerate}
We now present proofs for the individual steps:
\begin{claim}\label{clm:1}
When $\epsilon \geq \tilde \epsilon = \frac{2\alpha(1-p)}{(1-2\alpha)}$ and is in favour of group $A$, the system has a fixed point $(x_A, x_B)$ in the over-subscribed regime with $x_A = 2\alpha(p-\epsilon)+\epsilon$ and $x_B = 0$. 
\end{claim}
\begin{proof}
We will prove this claim by direct verification. Firstly, observe that $\left(2\alpha(p-\epsilon) + \epsilon, 0\right)$ is a point in the over-subscribed regime when $\epsilon \geq \tilde \epsilon$. This is because: 
\[
   2\alpha(p-\epsilon) + \epsilon + 0 = 2\alpha p + \epsilon(1-2\alpha) \geq 2\alpha p + 2\alpha (1-p) = 2\alpha. 
\]
Also, it can be seen that $\left(2\alpha(p-\epsilon) + \epsilon, 0\right)$ satisfies the fixed point equations: 
\[
    x_A = \frac{x_A (2\alpha)}{x_A + x_B}\cdot p + \left(1 -  \frac{x_A(2\alpha)}{x_A + x_B}\right)\cdot \epsilon; \quad x_B = \frac{x_B(2\alpha)}{x_A + x_B}\cdot p.
\]
This implies that $\left(2\alpha(p-\epsilon) + \epsilon, 0\right)$ is indeed a fixed point of the over-subscribed regime. 
\end{proof}

\begin{lem}\label{lem:1}
When $\epsilon \geq \tilde \epsilon$ and is in favour of group $A$, then for any starting point where group $A$ is better off, the system always reaches the fixed point $(2\alpha(p-\epsilon)+\epsilon, 0)$.
\end{lem}
\begin{proof}
We will do a case-by-case analysis, considering the location of the system at time $t$ (over-subscribed or under-subscribed regime) and then reasoning about the evolution in future time steps. WLOG, we assume that $x_A(t) > x_B(t)$. \\

\noindent
\textit{Subcase a):} Suppose, $(x_A(t), x_B(t))$ satisfies $x_A(t) + x_B(t) \geq 2\alpha$ (over-subscribed regime) with $x_A(t) > x_B(t)$. Then, for the next time step, we have: 
\[
       x_A(t+1) = \frac{x_A(t)}{x_A(t) + x_B(t)}\cdot (2\alpha)\cdot (p-\epsilon) + \epsilon, 
\]
\[
      x_B(t+1) = \frac{x_B(t)}{x_A(t) + x_B(t)}\cdot (2\alpha)\cdot p
\]
First, we claim that $x_A(t+1) + x_B(t+1) \geq 2 \alpha$. We will prove by contradiction. Suppose, if possible, $x_A(t+1) + x_B(t+1) < 2\alpha$. This implies: 
\begin{align*}
    x_A(t+1) + x_B(t+1) &= 2\alpha p + \epsilon \left(1 - \frac{2\alpha x_A(t)}{x_A(t) + x_B(t)} \right) < 2\alpha \\
    \implies & \epsilon \left(1 - \frac{2\alpha x_A(t)}{x_A(t) + x_B(t)} \right) < 2\alpha (1-p) \\
    \implies & \epsilon < \frac{2\alpha (1-p)}{\left(1 - \frac{2\alpha x_A(t)}{x_A(t) + x_B(t)} \right)} \leq \frac{2\alpha (1-p)}{(1-2\alpha)}.
\end{align*}
But this is a contradiction because $\epsilon \geq \frac{2\alpha(1-p)}{(1-2\alpha)}$. So we have shown that if the system starts off in the over-subscribed regime, it will continue to be there for all future time points. \\

Now, we make the following observations: 
\begin{enumerate}
    \item Observe that $x_B(t+1) = \frac{x_B(t)}{x_A(t) + x_B(t)}\cdot (2\alpha p)$ implies that $x_B(t+1) \leq p \cdot x_B(t)$ since $x_A(t) + x_B(t) \geq 2 \alpha$ by assumption. We have already argued that $x_A(t) + x_B(t) \geq 2\alpha \implies x_A(t+1) + x_B(t+1) \geq 2 \alpha$. Therefore, the same applies at time step $t+1$ implying that $x_B(t+2) \leq p\cdot x_B(t+1) \leq p^2 \cdot x_B(t)$. By induction, $x_B(t+n) \leq p^n \cdot x_B(t)$. Thus, $\lim_{n \to \infty}\{x_B(t+n)\} = 0$.   
    \item Since $x_A(t) + x_B(t) \geq 2 \alpha$, we know that for any $n \geq 1$, $x_A(t+n) + x_B(t+n) \geq 2 \alpha$. We have already shown that by choosing $n$ large enough, $x_B(t+n)$ can be made arbitrarily small. Therefore, 
    \[
           \frac{x_A(t+n)}{x_A(t+n) + x_B(t+n)} \to 1. 
    \]
    \item Using the last observation, we conclude that:
    \[
           \lim_{n \to \infty} x_A(t+n) = \lim_{n \to \infty} \frac{x_A(t+n-1)}{x_A(t+n-1) +x_B(t+n-1)}\cdot (2\alpha)(p-\epsilon) + \epsilon = 2\alpha (p-\epsilon) + \epsilon.
    \]
\end{enumerate}
This concludes the first part of the proof --- if the starting point lies in the over-subscribed regime, the system will always converge to the fixed point $(x_A, x_B)$ given by $\left(2\alpha(p-\epsilon)+\epsilon, 0\right)$. \\

\noindent
\textit{Subcase b):} Now, suppose that $(x_A(t), x_B(t))$ satisfies $x_A(t) + x_B(t) <  2\alpha$ (starting point is in the under-subscribed regime) with $x_A(t) > x_B(t)$. Then, we have the following: 
\begin{align*}
    x_A(t+1) &= x_A(t) \cdot p + \frac{(1-x_A(t))}{(2-x_A(t)-x_B(t))}\cdot (2\alpha - x_A(t) - x_B(t))\cdot p + \frac{2(1-\alpha)(1-x_A(t))}{(2-x_A(t) - x_B(t))}\cdot \epsilon, \\
    x_B(t+1) &= x_B(t) \cdot p + \frac{(1-x_B(t))}{(2-x_A(t)-x_B(t))}\cdot (2\alpha - x_A(t) - x_B(t))\cdot p
\end{align*}
Clearly, $x_A(t+1) > x_B(t+1)$. Our goal now is to show that $x_A(t+1) + x_B(t+1) \geq 2 \alpha$. If we can show this, the dynamics of the over-subscribed regime (subcase a) takes over and we have already shown that it converges to the desired fixed point. So if we can show that $x_A(t+1) + x_B(t+1) \geq 2 \alpha$, we are done. We will prove by contradiction. Suppose, $x_A(t+1) + x_B(t+1) <  2\alpha$. This implies: 
\begin{align*}
    x_A(t+1) + x_B(t+1) =& ~2\alpha p + \frac{2(1-\alpha)(1-x_A(t))}{(2-x_A(t) - x_B(t))}\cdot \epsilon < 2\alpha \\
    \implies & ~ \epsilon < \frac{2\alpha(1-p)}{(1-x_A(t))}\cdot \frac{(2-x_A(t) - x_B(t))}{2(1-\alpha)}. 
\end{align*}
But, we know that $\epsilon \geq \frac{2\alpha(1-p)}{(1-2\alpha)}$. Therefore, it must be that:
\begin{align*}
   & \frac{2\alpha(1-p)}{(1-2\alpha)} < \frac{2\alpha(1-p)}{(1-x_A(t))}\cdot \frac{(2-x_A(t) - x_B(t))}{2(1-\alpha)} \\
   \iff & (2-2\alpha)(1-x_A(t)) < (1-2\alpha)(2-x_A(t)-x_B(t))\\
   \iff & 2 - 2 x_A(t) - 2\alpha + 2\alpha x_A(t) < 2 - x_A(t) - x_B(t) - 4\alpha + 2 \alpha x_A(t) + 2\alpha x_B(t) \\
   \iff & 2\alpha - 2\alpha x_B(t) < x_A(t) - x_B(t).
\end{align*}
However, since $x_A(t) < 2\alpha - x_B(t)$, $x_A(t) - x_B(t) < 2\alpha - 2x_B(t)$. This implies, $2\alpha - 2\alpha x_B(t) < 2 - 2x_B(t) \implies 2x_B(t) < 2\alpha x_B(t)$. This is clearly a contradiction because $\alpha < \frac{1}{2}$. This concludes the proof. 
\end{proof}

\begin{lem}\label{lem:brouwer}
Suppose that $\alpha < \frac{1}{2}$. If the affinity parameter $\epsilon$ satisfies $0 < \epsilon < \tilde \epsilon = \frac{2\alpha(1-p)}{(1-2\alpha)}$ and is in favour of group $A$, then the system has a fixed point $(x_A, x_B)$ with $x_A > x_B$ and $x_A + x_B < 2\alpha$. 
\end{lem}
\begin{proof}
The proof will be using Brouwer's fixed point theorem. Consider the following set: 
\begin{align*}
    \mathcal{X} := \left\{ (x_A, x_B): 0 \leq x_A, x_B \leq 1,~ x_A + x_B \leq 2\alpha,~ x_A \geq x_B,~ x_A \geq 1 - \frac{2\alpha (1-p)}{\epsilon}  \right\}
\end{align*}
We make the following observations about $\mathcal{X}$:
\begin{itemize}
    \item $\mathcal{X} \neq \emptyset$. This follows because $\epsilon < \frac{2\alpha (1-p)}{(1-2\alpha)}$ which guarantees that $1 - \frac{2\alpha (1-p)}{\epsilon} <  2\alpha$. 
    \item Additionally, $\mathcal{X}$ is a convex set because it is defined entirely by linear inequalities. 
    \item $\mathcal{X}$ is also a compact set (it is both closed and bounded).  
\end{itemize}
Now, consider the mapping $\mathcal{T}: \mathcal{X} \to \mathbb{R}^2_{\geq 0}$ such that $(y_A, y_B) = \mathcal{T}(x_A, x_B)$ for $(x_A, x_B) \in \mathcal{X}$.
\begin{align*}
    y_A &= x_A \cdot p + \frac{(1-x_A)(2\alpha -x_A - x_B)}{(2-x_A - x_B)}\cdot p + \frac{2(1-\alpha)(1-x_A)}{(2-x_A - x_B)}\cdot \epsilon\\
    y_B &= x_B \cdot p + \frac{(1-x_B)(2\alpha -x_A - x_B)}{(2-x_A - x_B)}\cdot p
\end{align*}
Clearly, $\mathcal{T}(\cdot)$ is a continuous mapping. We will now show that $\mathcal{T}$ maps from $\mathcal{X}$ to itself, i.e., $(y_A, y_B) \in \mathcal{X}$. It can be easily verified that $0 \leq y_A, y_B \leq 1$. Now, \begin{align*}
    y_A - y_B &= (x_A - x_B)\cdot p - \frac{(x_A - x_B)(2\alpha - x_A - x_B)}{(2-x_A -x_B)}\cdot p + \frac{2(1-\alpha)(1-x_A)}{(2-x_A - x_B)}\cdot \epsilon \\
    &= \frac{2(1-\alpha)(x_A - x_B)}{(2-x_A -x_B)}\cdot p + \frac{2(1-\alpha)(1-x_A)}{(2-x_A - x_B)}\cdot \epsilon \\
    &> 0. \quad \text{(since $x_A \geq x_B$)}
\end{align*}
Further, 
\begin{align*}
    y_A + y_B &= 2\alpha p + \frac{2(1-\alpha)(1-x_A)}{(2-x_A-x_B)}\cdot \epsilon \\
    &\leq 2\alpha p + \frac{2(1-\alpha)}{(2-x_A-x_B)}\cdot 2\alpha (1-p) \quad \text{($\because~ x_A \geq 1-\frac{2\alpha(1-p)}{\epsilon} \implies (1-x_A)\epsilon \leq 2\alpha (1-p)$)} \\
    &\leq 2\alpha p + 2\alpha (1-p) \quad \text{($\because~x_A + x_B \leq 2\alpha \implies 2-x_A-x_B \geq 2(1-\alpha) \implies \frac{2(1-\alpha)}{(2-x_A-x_B)} \leq 1$)} \\
    &= 2\alpha. 
\end{align*}
Finally, we have to show that $y_A \geq 1 - \frac{2\alpha(1-p)}{\epsilon}$. We have: 
\begin{align*}
    y_A &= x_A \cdot p + \frac{(1-x_A)(2\alpha -x_A - x_B)}{(2-x_A - x_B)}\cdot p + \frac{2(1-\alpha)(1-x_A)}{(2-x_A - x_B)}\cdot \epsilon\\
    &= x_A \cdot p + \frac{(2\alpha -x_A-x_B)}{(2-x_A-x_B)}\cdot p - (x_A \cdot p) \cdot  \frac{(2\alpha -x_A-x_B)}{(2-x_A-x_B)} + \frac{2(1-\alpha)(1-x_A)}{(2-x_A - x_B)}\cdot \epsilon\\
    &= \frac{2(1-\alpha)}{(2-x_A-x_B)}\cdot (x_A \cdot p) + \frac{(2\alpha -x_A-x_B)}{(2-x_A-x_B)}\cdot p + \frac{2(1-\alpha)(1-x_A)}{(2-x_A - x_B)}\cdot \epsilon\\
    &= \frac{2(1-\alpha)}{(2-x_A-x_B)}\left(x_A \cdot p + (1-x_A)\cdot \epsilon \right) + \frac{(2\alpha -x_A-x_B)}{(2-x_A-x_B)}\cdot p \\
    &= \frac{2(1-\alpha)}{(2-x_A-x_B)}\left(x_A \cdot (p-\epsilon) + \epsilon \right) + \frac{(2\alpha -x_A-x_B)}{(2-x_A-x_B)}\cdot p \\
    &\geq \frac{2(1-\alpha)}{(2-x_A-x_B)}\left( (1-\frac{2\alpha(1-p)}{\epsilon}) \cdot (p-\epsilon) + \epsilon \right) + \frac{(2\alpha -x_A-x_B)}{(2-x_A-x_B)}\cdot p \\
    &= \frac{2(1-\alpha)}{(2-x_A-x_B)}\left( p-\frac{2\alpha(1-p)(p-\epsilon)}{\epsilon}  \right) + \frac{(2\alpha -x_A-x_B)}{(2-x_A-x_B)}\cdot p \\
    &= \frac{2(1-\alpha)p}{(2-x_A-x_B)} + \frac{(2\alpha -x_A-x_B)p}{(2-x_A-x_B)} - \frac{4\alpha(1-\alpha)(1-p)(p-\epsilon)}{\epsilon(2-x_A-x_B)}\\
    &= p - \frac{2(1-\alpha)}{(2-x_A-x_B)}\cdot \frac{2\alpha(1-p)(p-\epsilon)}{\epsilon} \\
    &\geq p - \frac{2\alpha(1-p)(p-\epsilon)}{\epsilon} \quad \text{(using $\frac{2(1-\alpha)}{(2-x_A-x_B)} \leq 1$)}\\
    &= p \left(1 - \frac{2\alpha(1-p)}{\epsilon} \right) + (1-p)(2\alpha) \\
    &> 1 - \frac{2\alpha(1-p)}{\epsilon}.
    \end{align*}
The last step follows from the fact that $1 - \frac{2\alpha(1-p)}{\epsilon} <  2\alpha$ due to $\epsilon < \frac{2\alpha(1-p)}{(1-2\alpha)}$ and any convex combination of two terms is always greater than the minimum of the terms. This concludes the proof that $(y_A, y_B) \in \mathcal{X}$. \\

Thus, $\mathcal{T}$ is a continuous mapping which maps from $\mathcal{X}$ to itself with $\mathcal{X}$ being non-empty, convex and compact. Therefore, by Brouwer's fixed point theorem, there must exist a fixed point of $\mathcal{T}$ in $\mathcal{X}$. However, we are not done yet -- in order to show that the fixed point $(x_A, x_B)$ satisfies both $x_A > x_B$ and $x_A + x_B < 2\alpha$, we still need to argue that it cannot satisfy $x_A + x_B = 2\alpha$ or $x_A = x_B$. Each of these cases is easy to verify. Firstly, if $x_A = x_B$ at the fixed point, we must have $\frac{(1-x_A)}{(2-x_A-x_B)} = 0$ which implies $x_A = x_B = 1$ which violates $x_A + x_B \leq 2\alpha$. Similarly, if $x_A + x_B = 2\alpha$, we must have $x_B = 0$ implying $x_A = 2\alpha$, but then it does not satisfy the first fixed point equation. Therefore, the fixed point must satisfy $x_A > x_B$ and $x_A + x_B < 2\alpha$ and is, thus, a fixed point of the under-subscribed regime. This concludes the proof of existence.   
\end{proof}

\begin{lem}\label{lem:fp_under}
Suppose that $\alpha < \frac{1}{2}$ and $0 < \epsilon < \tilde \epsilon$ in favour of group $A$. Then the system has \textbf{only one} fixed point $(x_A, x_B)$ which satisfies both $x_A > x_B$ and $x_A + x_B < 2\alpha$. Further, we have: 
\[
               x_A > 1 - \min\left(\frac{2\alpha (1-p)}{\epsilon}, \frac{(1-\alpha p)^2}{2(1-\alpha)\epsilon} \right),
\]
and the separation between the groups at the fixed point is strictly greater than\\ $\max\left(0, 2(1-\alpha) - \min\left( \frac{4\alpha(1-p)}{\epsilon}, \frac{(1-\alpha p)^2}{(1-\alpha)\epsilon} \right)\right)$.
\end{lem}
\begin{proof}
Since $\alpha < \frac{1}{2}$ and $0 < \epsilon < \tilde \epsilon$, by the previous lemma, there does exist at least one fixed point $(x_A, x_B)$ with $x_A > x_B$ and $x_A + x_B < 2\alpha$. Recall, we have already shown that at any fixed point in the under-subscribed regime, $x_A$ and $x_B$ must satisfy: 
\[
       (x_A + x_B) = 2\alpha p + \frac{2(1-\alpha)(1-x_A)}{(2-x_A-x_B)}\cdot \epsilon.
\]
This implies that $(x_A + x_B)$ must be a solution to the following quadratic equation: 
\[
        g(y) = y^2 - (2+2\alpha p)y + 4\alpha p + 2(1-\alpha)(1-x_A)\epsilon = 0. 
\]
Since we know that at least one fixed point exists, $g(y) = 0$ must have at least one real root which implies that it must have both real roots (since a quadratic equation cannot have odd number of real roots). Therefore, $(x_A + x_B)$ must be from the set: 
\[
        \left\{ (1+\alpha p) - \sqrt{(1-\alpha p)^2 - 2(1-\alpha)(1-x_A)\epsilon},  (1+\alpha p) + \sqrt{(1-\alpha p)^2 - 2(1-\alpha)(1-x_A)\epsilon} \right\}
\]
Now, since $\alpha < \frac{1}{2}$, $1+\alpha p > 2\alpha$. Therefore, only one of the above solutions (the first one) is a candidate value for $x_A + x_B$ at a fixed point of the under-subscribed regime. This concludes the proof that there can be \textbf{only one such fixed point} (because once $x_A + x_B$ is determined uniquely, $x_A$ and $x_B$ are also characterized uniquely).  \\

Now, note that $2\alpha p < x_A + x_B < 2\alpha < 1 + \alpha p$. Further, we have the following: 
\begin{enumerate}
    \item $g(x_A + x_B)$ = 0 (self-explanatory).
    \item We can verify that $g(2\alpha p) = 2(1-\alpha)(1-x_A) > 0$.
    \item $g(y)$ attains its minimum value at $y = 1+\alpha p$. Since $g(y) = 0$ has distinct real roots, $g(1+\alpha p)$ must be $< 0$ which implies the discriminant is $> 0$ or equivalently, $(1-\alpha p)^2 > 2(1-\alpha)(1-x_A)\epsilon$ or $x_A > 1 - \frac{(1-\alpha p)^2}{2(1-\alpha)\epsilon}$.
\end{enumerate}
Since $g(\cdot)$ is clearly a continuous function, by the intermediate value theorem, we must also have $g(2\alpha) < 0$. 
\begin{align*}
    g(2\alpha) < 0 &\implies 4\alpha^2 - (2+2\alpha p)(2\alpha) + 4\alpha p + 2(1-\alpha)(1-x_A)\epsilon < 0 \\
    &\implies 4\alpha^2 - 4\alpha - 4\alpha^2 p + 4\alpha p + 2(1-\alpha)(1-x_A)\epsilon < 0 \\
    &\implies -4\alpha(1-\alpha)(1-p) + 2(1-\alpha)(1-x_A)\epsilon < 0 \\
    &\implies 2(1-\alpha) \left[ -2\alpha(1-p) + (1-x_A)\epsilon \right] < 0 \\
    &\implies (1-x_A)\epsilon < 2\alpha (1-p) \\
    &\implies x_A > 1 - \frac{2\alpha(1-p)}{\epsilon}. 
\end{align*}
Therefore, combining both results, we obtain: 
\[
          x_A > \max \left[ 1 - \frac{2\alpha(1-p)}{\epsilon}, 1 -  \frac{(1-\alpha p)^2}{2(1-\alpha)\epsilon}  \right] = 1 - \min\left(\frac{2\alpha(1-p)}{\epsilon}, \frac{(1-\alpha p)^2}{2(1-\alpha)\epsilon} \right).
\]
Finally, at the fixed point, the separation $\Delta$ is given by: 
\begin{align*}
    \Delta = x_A - x_B &= x_A - (x_A + x_B -x_A) \\
    &= 2x_A - (x_A + x_B) \\
    &> 2 - \min\left( \frac{4\alpha(1-p)}{\epsilon}, \frac{(1-\alpha p)^2}{(1-\alpha)\epsilon} \right) - 2\alpha \quad \text{(since $x_A + x_B < 2\alpha$)} \\
    &= 2(1-\alpha) - \min\left( \frac{4\alpha(1-p)}{\epsilon}, \frac{(1-\alpha p)^2}{(1-\alpha)\epsilon} \right).
\end{align*}
This lower bound is loose and may be negative which is why, we take the max with $0$ (since $x_A > x_B$, $\Delta > 0$ by default). This concludes the proof. 
\end{proof}

\begin{lem}\label{lem:unique}
When $0 < \epsilon < \tilde \epsilon$ and is in favour of group $A$, then for any arbitrary starting point where group $A$ is better off, the system will always reach the unique fixed point which satisfies $x_A > x_B$ and $x_A + x_B < 2\alpha$.
\end{lem}
\begin{proof}
We will organize the proof as follows: 
\begin{figure}[!ht]
    \centering
    \includegraphics[width=0.5\linewidth]{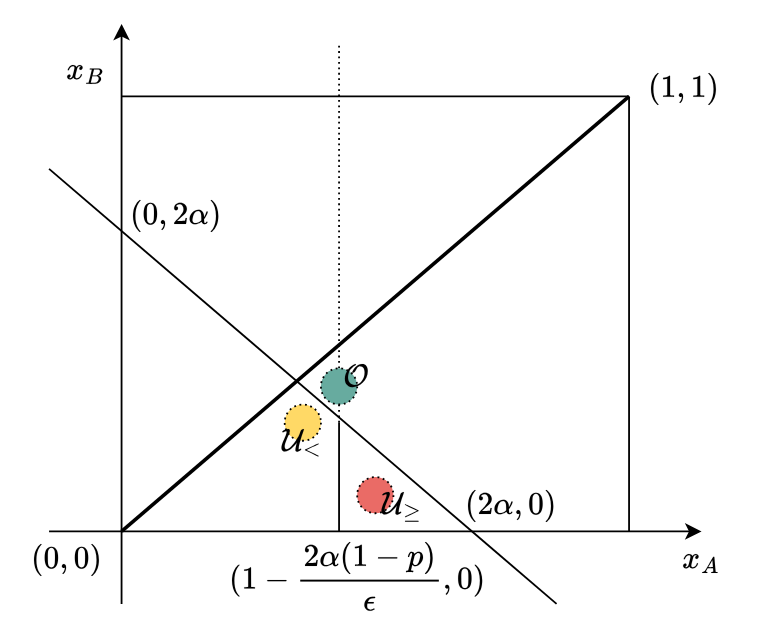}
    \caption{Feasible space with sub-regions}
    \label{fig:regions}
\end{figure}
\begin{enumerate}
    \item As shown in Figure~\ref{fig:regions}, any feasible starting point can belong in one of the $3$ sub-regions: $\mathcal{O}$, $\mathcal{U}_{<}$ and $\mathcal{U}_{\geq}$. The $3$ sub-regions are defined as follows:
    \begin{align*}
        \mathcal{O} &:= \left\{ (x,y) \in \mathbb{R}^2: 0 \leq x,y \leq 1,~x > y,~ x+y \geq 2\alpha  \right\};\\
        \mathcal{U}_{<} &:= \left\{ (x, y) \in \mathbb{R}^2: 0 \leq x,y \leq 1,~x > y, ~x+y < 2\alpha, x < 1 - \frac{2\alpha(1-p)}{\epsilon} \right\}; \\
        \mathcal{U}_{\geq} &:= \left\{ (x, y) \in \mathbb{R}^2: 0 \leq x,y \leq 1,~x > y,~x+y < 2\alpha, x \geq 1 - \frac{2\alpha(1-p)}{\epsilon} \right\};
    \end{align*}
    \item We will show that irrespective of the starting point (either $\mathcal{O}$ or $\mathcal{U}_{<}$), the system will always eventually reach $U_{\geq}$. Recall that we have already shown (in the proof for Lemma~\ref{lem:brouwer}) that once the system reaches $\mathcal{U}_{\geq}$, it will continue to remain there at all future time points. 
    \item Finally, we will show that any starting point in $\mathcal{U}_{\geq}$ always converges to the only fixed point in $U_{\geq}$. 
\end{enumerate}

\noindent
\textit{Starting at any point in $\mathcal{O}$, the system always reaches $\mathcal{U}_{\geq}$.} First we will show that starting at any point in $\mathcal{O}$, the system \textbf{cannot} stay in $\mathcal{O}$ indefinitely. We will prove by contradiction. Suppose, there exists a starting point in $\mathcal{O}$ such that the system stays in $\mathcal{O}$ for all future time steps. In that case, the system evolves according to the dynamics of the over-subscribed regime: 
\begin{align*}
    x_A(t+1) &= \frac{x_A(t)}{x_A(t) + x_B(t)}\cdot (2\alpha p) + \left(1 - \frac{2\alpha x_A(t)}{x_A(t) + x_B(t)} \right)\cdot \epsilon,\\
    x_B(t+1) &= \frac{x_B(t)}{x_A(t) + x_B(t)}\cdot (2\alpha p). 
\end{align*}
Since $x_A(t) + x_B(t) \geq 2 \alpha$ for all $t \geq 0$ (by assumption), we have $x_B(t+1) \leq p \cdot x_B(t)$ for all $t \geq 0$. This implies that $\lim_{t \to \infty} x_B(t) = 0$. But since $x_A(t+1) + x_B(t+1) = \frac{2\alpha x_A(t)}{x_A(t) + x_B(t)}(p-\epsilon) + \epsilon$, this implies that $\lim_{t \to \infty} x_A(t) + x_B(t) = 2\alpha (p-\epsilon) + \epsilon < 2\alpha$ since $\epsilon < \frac{2\alpha(1-p)}{(1-2\alpha)}$. But this is a contradiction because a system for which $x_A(t) + x_B(t) \geq 2\alpha~\forall~t$ cannot have a fixed point $(x_A, x_B)$ with $x_A + x_B < 2\alpha$. This concludes the proof that starting in $\mathcal{O}$, the system cannot remain in $\mathcal{O}$ indefinitely and therefore, must transition to $\mathcal{U}_{<}$ or $\mathcal{U}_{\geq}$. \\

We will now show that the system always transitions from $\mathcal{O}$ to $\mathcal{U}_{\geq}$. Suppose, the transition happens at time $t+1$, i.e., the system is in $\mathcal{O}$ at time $t$ and then in $\mathcal{U}_{<} \bigcup \mathcal{U}_{\geq}$ at time $t+1$. Then we have the following conditions: 
\begin{align*}
    &x_A(t) + x_B(t) \geq 2\alpha; \\
    &x_A(t+1) = \frac{x_A(t)}{x_A(t) + x_B(t)}\cdot (2\alpha p) + \left(1 - \frac{2\alpha x_A(t)}{x_A(t) + x_B(t)} \right)\cdot \epsilon; \quad x_B(t+1) = \frac{x_B(t)}{x_A(t) + x_B(t)}\cdot (2\alpha p); \\
    &x_A(t+1) + x_B(t+1) < 2\alpha. 
\end{align*}
Now, 
\begin{align*}
    x_A(t+1) + x_B(t+1) &= 2\alpha p + \left(1 - \frac{2\alpha x_A(t)}{x_A(t) + x_B(t)} \right)\cdot \epsilon < 2\alpha \\
    &\implies \frac{2\alpha x_A(t)}{x_A(t) + x_B(t)} > 1 - \frac{2\alpha (1-p)}{\epsilon}.
\end{align*}
In order to show that $(x_A(t+1), x_B(t+1)) \in \mathcal{U}_{\geq}$, it suffices to show that $x_A(t+1) \geq 1 - \frac{2\alpha (1-p)}{\epsilon}$. We will prove by contradiction. Suppose, this is not true, i.e.,
\begin{align*}
    x_A(t+1) < 1 - \frac{2\alpha (1-p)}{\epsilon} &\implies \frac{2\alpha x_A(t)}{x_A(t) + x_B(t)}\cdot (p-\epsilon) + \epsilon < 1- \frac{2\alpha (1-p)}{\epsilon} \\
    &\implies \left(1 - \frac{2\alpha (1-p)}{\epsilon} \right)(p-\epsilon) + \epsilon < 1 - \frac{2\alpha (1-p)}{\epsilon}  \\
    &\implies \epsilon < \left(1 - \frac{2\alpha (1-p)}{\epsilon} \right)(1-p+\epsilon) \\
    &\implies \epsilon < 1-p+\epsilon - \frac{2\alpha(1-p)^2}{\epsilon} - 2\alpha(1-p) \\
    &\implies 0 < (1-p)\left(1 - 2\alpha - \frac{2\alpha(1-p)}{\epsilon} \right) \\
    &\implies \frac{2\alpha (1-p)}{\epsilon} < (1-2\alpha) \\
    &\implies \epsilon > \frac{2\alpha(1-p)}{(1-2\alpha)},
\end{align*}
which is a contradiction. This completes the proof. \\

\noindent
\textit{Starting at any point in $\mathcal{U}_{<}$, the system cannot remain in $\mathcal{U}_{<}$ indefinitely.} We will prove by contradiction. Suppose, there exists a point in $\mathcal{U}_{<}$ starting at which the system remains in $\mathcal{U}_{<}$ indefinitely. Note that since $\mathcal{U}_{<}$ is a bounded space, this can happen in only two ways: i) either the system converges to or approaches a fixed point, or ii) it cycles. 

First, we will argue that convergence to a fixed point is not possible. Clearly, the system cannot converge to a fixed point in $\mathcal{U}_{<}$ because $\mathcal{U}_{<}$ does not have a fixed point. The other possibility is that since $\mathcal{U}_{<}$ is not closed, the system can approach a fixed point in $cl(\mathcal{U}_{<})\setminus \mathcal{U}_{<}$ while remaining in $\mathcal{U}_{<}$ indefinitely. But, we have established that the unique fixed point of the system does not satisfy $x_A + x_B = 2\alpha$ or $x_A = 1 - \frac{2\alpha(1-p)}{\epsilon}$ and thus cannot belong in $cl(\mathcal{U}_{<})\setminus \mathcal{U}_{<}$. Therefore, the system cannot converge to or approach a fixed point while remaining in $\mathcal{U}_{<}$ indefinitely. 

Now, we will prove that the system also cannot cycle in $\mathcal{U}_{<}$. We will prove this by showing that for any $(x_A(t), x_B(t)) \in \mathcal{U}_{<}$, it must be that $x_A(t+1) + x_B(t+1) \geq x_A(t) + x_B(t)$. Monotonicity guarantees that there cannot be a cycle. We will again do a proof by contradiction. Suppose, there exists a time step $t$ such that $x_A(t+1) + x_B(t+1) < x_A(t) + x_B(t)$. This implies: 
\begin{align*}
    &2\alpha p + \frac{2(1-\alpha)(1-x_A(t))}{(2-x_A(t) - x_B(t))}\cdot \epsilon < x_A(t) + x_B(t) \\
    \implies & (x_A(t) + x_B(t))^2 - (2+2\alpha p)(x_A(t) + x_B(t)) + 4\alpha p + 2(1-\alpha)(1-x_A(t))\epsilon < 0 \\
    \implies & (x_A(t) + x_B(t))^2 - (2+2\alpha p)(x_A(t) + x_B(t)) + 4\alpha p + 4\alpha(1-\alpha)(1-p) < 0.
\end{align*}
The last step follows because $x_A(t) < 1-\frac{2\alpha(1-p)}{\epsilon} \implies (1-x_A(t))\epsilon) > 2\alpha(1-p)$. Now, the quadratic polynomial can take negative values if the discriminant $D > 0$ and $x_A(t) + x_B(t)$ lies between the two roots, i.e., $1+\alpha p - \sqrt{D} < x_A(t) + x_B(t) < 1+\alpha p + \sqrt{D}$. The discriminant $D$ is given by: 
\begin{align*}
    D &= (1+\alpha p)^2 - 4\alpha p - 4 \alpha(1-\alpha)(1-p)\\
    &= (1+\alpha p)^2 - 4\alpha (1-\alpha) - 4\alpha^2 p \\
    &= (1+\alpha p)^2 + 4\alpha^2 - 4\alpha(1+\alpha p)\\
    &= (1+\alpha p - 2\alpha)^2. 
\end{align*}
Therefore, the above condition becomes: 
\[
      2\alpha  < x_A(t) + x_B(t) < 2(1+\alpha p) - 2\alpha.
\]
But this is a contradiction because $(x_A(t), x_B(t)) \in \mathcal{U}_{<}$ and therefore, $x_A(t) + x_B(t) < 2\alpha$. This concludes the proof. \\

It should be clear that the above result implies that starting at any point in $\mathcal{U}_{<}$, the system either transitions directly to $\mathcal{U}_{\geq}$ (in which case we are done), or transitions first to $\mathcal{O}$ and then must transition to $\mathcal{U}_{\geq}$ again (as we have shown before). Thus, combining both statements above, irrespective of the starting point, the system always reaches $\mathcal{U}_{\geq}$.  \\

\noindent
\textit{Starting at any point in $\mathcal{U}_{\geq}$, the system always converges to the unique fixed point in $\mathcal{U}_{\geq}$.} In order to complete this proof, we will first show that for any starting point in $\mathcal{U}_{\geq}$, the system evolves in a way that $x_A(t) + x_B(t)$ is decreasing in $t$, i.e., $x_A(t+1) + x_B(t+1) \leq x_A(t) + x_B(t)$ for all $t$. We will prove by contradiction. Suppose, there exists a time step $t \geq 1$ such that the following hold simultaneously: 
\begin{align*}
    &x_A(t) + x_B(t) < 2\alpha, \quad x_A(t) \geq 1-\frac{2\alpha(1-p)}{\epsilon}; \\
    &x_A(t) + x_B(t) < x_A(t+1) + x_B(t+1) < 2\alpha. 
\end{align*}
This implies:
\begin{align*}
    &x_A(t+1) + x_B(t+1) \\
    &= 2\alpha p + \frac{2(1-\alpha)(1-x_A(t))}{(2-x_A(t)-x_B(t))}\cdot \epsilon > x_A(t) + x_B(t) \\
    &\implies \left(x_A(t) + x_B(t)\right)^2 - (2+2\alpha p)(x_A(t) + x_B(t)) + 4\alpha p + 2(1-\alpha)(1-x_A(t))\epsilon > 0\\
    &\stackrel{(i)}{\implies} \left(x_A(t) + x_B(t)\right)^2 - (2+2\alpha p)(x_A(t) + x_B(t)) + 4\alpha p + 4\alpha(1-\alpha)(1-p) > 0\\
    &\implies (1+\alpha p)^2 < 4\alpha p + 4\alpha (1-\alpha)(1-p) = 4\alpha (1-\alpha + \alpha p) \quad \text{(no real roots, discriminant $<0$)} \\
    &\implies (1+\alpha p)^2 < 4\alpha(1-\alpha) + 4\alpha^2 p < 1 + 4\alpha^2 p \quad \text{(since $\alpha(1-\alpha) < \frac{1}{4}$ for $\alpha \in (0, \frac{1}{2})$)} \\
    &\implies 1 + 2\alpha p + \alpha^2 p^2 < 1 + 4\alpha^2 p \\
    &\implies 2+\alpha p < 4 \alpha. 
\end{align*}
Step $(i)$ above follows because $x_A(t) \geq 1-\frac{2\alpha(1-p)}{\epsilon} \implies (1-x_A(t))\epsilon \leq 2\alpha(1-p)$. The last step is a contradiction because the RHS $< 2$ since $\alpha < \frac{1}{2}$.  This proves that for any starting point in $\mathcal{U}_{\geq}$, $x_A(t) + x_B(t)$ is decreasing in $t$. We also know that $x_A(t) + x_B(t) \geq 2\alpha p$ for all $t$. Therefore, by the monotone convergence theorem, $x_A(t) + x_B(t)$ must converge as $t \to \infty$. But once $x_A(t) + x_B(t)$ converges, $x_A(t)$ must converge as well and so must $x_B(t)$. Thus, $\left\{(x_A(t), x_B(t))\right\}_{t \geq 0}$ converges to some fixed point (which must correspond to a fixed point in $\mathcal{U}_{\geq}$). But we know that there is a unique fixed point of the system in $\mathcal{U}_{\geq}$. Therefore, they must be the same fixed point. This concludes the proof.  
\end{proof}

\section{Extension to the Setting with Unequal Population Sizes}\label{app:unequal}

Let $N_A$ and $N_B$ be the population sizes of groups $A$ and $B$ with $N_A \neq N_B$. We define $\lambda_A = \frac{N_A}{N_A + N_B}$ and $\lambda_B = \frac{N_B}{N_A + N_B}$, i.e., $\lambda_A + \lambda_B = 1$. As earlier, let $X_A(t)$ and $X_B(t)$ denote the fraction of high types in groups $A$ and $B$ respectively at time $t$. \\

\noindent 
Based on our new notation, the total college capacity $C = (N_A + N_B)\alpha$. Also, our new definitions for the over- and under-subscribed regimes are as follows: 
\[
    \text{Over-subscribed: } N_A X_A(t) + N_B X_B(t) \geq (N_A + N_B)\alpha \iff 2\lambda_A X_A(t) + 2\lambda_B X_B(t) \geq 2\alpha.  
\]
\[
    \text{Under-subscribed: } N_A X_A(t) + N_B X_B(t) < (N_A + N_B)\alpha \iff 2\lambda_A X_A(t) + 2\lambda_B X_B(t) < 2\alpha. 
\]
Henceforth, we will use the following substitution: $Y_A(t) = 2\lambda_A X_A(t)$ and $Y_B(t) = 2\lambda_B X_B(t)$, whenever possible, to ease the exposition. 

\subsection{Revised Allocation Rule \& Fixed Point Equations}
Under the new setting, we construct the revised allocation rule $\mathcal{A}$ which is simultaneously meritocratic, fair and efficient. Once again, it is unique (by construction). Let $L_i(t+1)$ and $H_i(t+1)$ indicate the number of seats allocated at time $t+1$ to low type and high type individuals of group $i$ respectively for $i \in \{A, B\}$. 

\paragraph{Over-subscribed regime.} In the over-subscribed regime, no low types can go to college (meritocracy) as there are not enough seats for all high types in the population. Therefore, seats will be allocated only to high types. To guarantee fairness (through Equalized Odds), we must have: 
\begin{align*}
    \frac{H_A(t+1)}{N_A X_A(t)} = \frac{H_B(t+1)}{N_B X_B(t)} = \frac{H_A(t+1) + H_B(t+1) }{N_A X_A(t) + N_B X_B(t)}.
\end{align*}
For efficiency, all seats must be allocated, i.e., $H_A(t+1) + H_B(t+1) = (N_A + N_B)\alpha$. This leads to the following allocation rule in the over-subscribed regime:
\begin{align*}
      L_i(t+1) &= 0, \quad \forall~ i\in \{A, B\}, \\
      H_i(t+1) &= \frac{(N_A + N_B)\alpha}{N_A X_A(t) + N_B X_B(t)}\cdot N_i X_i(t) = \frac{(2\alpha)N_i X_i(t)}{2\lambda_A X_A(t) + 2\lambda_B X_B(t)}, \quad \forall~ i\in \{A, B\}. 
\end{align*}

We can now derive the fixed point equations for the over-subscribed regime. For all $i \in \{A, B\}$
\begin{align*}
   &N_i x_i = \frac{(2\alpha)N_i x_i}{2\lambda_A x_A + 2\lambda_B x_B}\cdot p + \left(N_i x_i - \frac{(2\alpha)N_i x_i}{2\lambda_A x_A + 2\lambda_B x_B} \right)\cdot q \\
   \iff &(2\lambda_i x_i) = \frac{(2\alpha)(2\lambda_i x_i)}{2\lambda_A x_A + 2\lambda_B x_B}\cdot p + \left( 2\lambda_i x_i - \frac{(2\alpha)(2\lambda_i x_i)}{2\lambda_A x_A + 2\lambda_B x_B} \right) \cdot q  \\
   \iff &y_i = \frac{(2\alpha)y_i}{y_A + y_B}\cdot p + \left(y_i - \frac{(2\alpha)y_i}{y_A + y_B} \right)\cdot q 
\end{align*}

\paragraph{Under-subscribed regime.} For the under-subscribed regime, meritocracy dictates that all high types individuals in the population be admitted. The residual seats given by $R = (N_A + N_B)\alpha - N_A X_A(t) - N_B X_B(t)$ have to be allocated among the low types individual from different groups. By the principle of equalized odds, we must have: 
\begin{align*}
    \frac{L_A(t+1)}{N_A(1 - X_A(t))} = \frac{L_B(t+1)}{N_B (1-X_B(t))} &= \frac{L_A(t+1) + L_B(t+1) }{N_A + N_B - N_A X_A(t) - N_B X_B(t)} \\
    &= \frac{(N_A + N_B)\alpha - N_A X_A(t) - N_B X_B(t)}{N_A + N_B - N_A X_A(t) - N_B X_B(t)} \\
    &= \frac{2\alpha - 2\lambda_A X_A(t) - 2\lambda_B X_B(t)}{2 - 2\lambda_A X_A(t) - 2\lambda_B X_B(t)}.
\end{align*}
This leads to the following allocation rule for the under-subscribed regime: 
\begin{align*}
    L_i(t+1) &=  \left(\frac{2\alpha - 2\lambda_A X_A(t) - 2\lambda_B X_B(t)}{2 - 2\lambda_A X_A(t) - 2\lambda_B X_B(t)}\right)\cdot N_i(1-X_i(t)), \quad i \in \{A, B\}\\
    H_i(t+1) &= N_i X_i(t), \quad i \in \{A, B\}.
\end{align*}

Therefore, the fixed point equations can be constructed as follows, for all $i \in \{A, B\}$: 
\begin{align*}
    &N_i x_i = \left(N_i x_i\right) \cdot p + \left( \frac{2\alpha - 2 \lambda_A x_A - 2\lambda_B x_B}{2-2\lambda_A x_A - 2\lambda_B x_B}  \right)\cdot N_i(1-x_i) \cdot p \\
    \iff & (2\lambda_i x_i) = (2\lambda_i x_i)\cdot p + \left( \frac{2\alpha - 2 \lambda_A x_A - 2\lambda_B x_B}{2-2\lambda_A x_A - 2\lambda_B x_B}  \right)\cdot 2\lambda_i(1-x_i) \cdot p \\
    \iff & y_i = y_i \cdot p + \left(\frac{2\alpha - y_A - y_B}{2-y_A - y_B} \right)\cdot (2\lambda_i - y_i)\cdot p.
\end{align*}

\subsection{Fixed Points in the EA regime}
The analysis here is exactly identical to the equal population case. We first show that there is no fixed point in the over-subscribed regime (by contradiction). In the under-subscribed regime, we have: 
\begin{align*}
    y_A &= y_A \cdot p + \left(\frac{2\alpha - y_A - y_B}{2-y_A - y_B} \right)\cdot (2\lambda_A - y_A)\cdot p, \\
    y_B &= y_B \cdot p + \left(\frac{2\alpha - y_A - y_B}{2-y_A - y_B} \right)\cdot (2\lambda_B - y_B)\cdot p.
\end{align*}

Adding the equations together, we obtain: 
\[
   y_A + y_B = (y_A + y_B)\cdot p + (2\alpha - y_A - y_B)\cdot p = 2\alpha p.
\]
Plugging back and solving, we obtain $y_A = (2\lambda_A)\alpha p$ and $y_B = (2\lambda_B)\alpha p$ which implies that $x_A = x_B = \alpha p$ at the fixed point. Using an identical analysis as earlier, we can show that all starting points must converge at this unique fixed point in the under-subscribed regime. 

\subsection{Fixed Points in the AA regime}
Using an identical analysis as the equal population case, we can derive similar results about the fixed points of the system even in the AA regime.

\section{Additional Experiments}\label{app:exp}

\subsection{Mean time for return to parity: Dependence on $p$ and $q$}
We complement our discussion in Section~\ref{sub:equal_adv} (Return to Parity) with Figures~\ref{fig:parity_q} and \ref{fig:parity_p}. These plots demonstrate how the average time for the system to reach $\eta$-parity depends on $q$ and $p$ respectively. Our target level of $\eta$ is equal to $0.01$ and in both figures, the system starts off at $(X_A(0), X_B(0)) = (0.9, 0.1)$ with an initial separation $|\Delta(0)| = 0.8$. On the left, we observe that there is no clear trend for how the average time to parity depends on $q$, but it is clear that the overall effect of $q$ is not that significant --- varying $q$ over the range $[0, 0.5]$ leads to a very small change in the mean time to parity. Note that our starting point is in the over-subscribed regime where $q$ does affect the dynamics. This implies that for starting points in the under-subscribed regime, the dependence on $q$ will be even more negligible.  

On the right, however, we find that the mean time to parity increases sharply with the success probability $p$. By varying $p \in [0.7, 0.96]$, the mean time to parity grows almost $7$-fold. As described earlier, this is due to the fact that as $p$ approaches $1$, college outcomes become more and more deterministic and any initial/accumulated separation is hard to correct. 
\begin{figure*}[ht]
    \centering
    \begin{subfigure}[b]{0.48\textwidth}
        \centering
        \includegraphics[height=6cm]{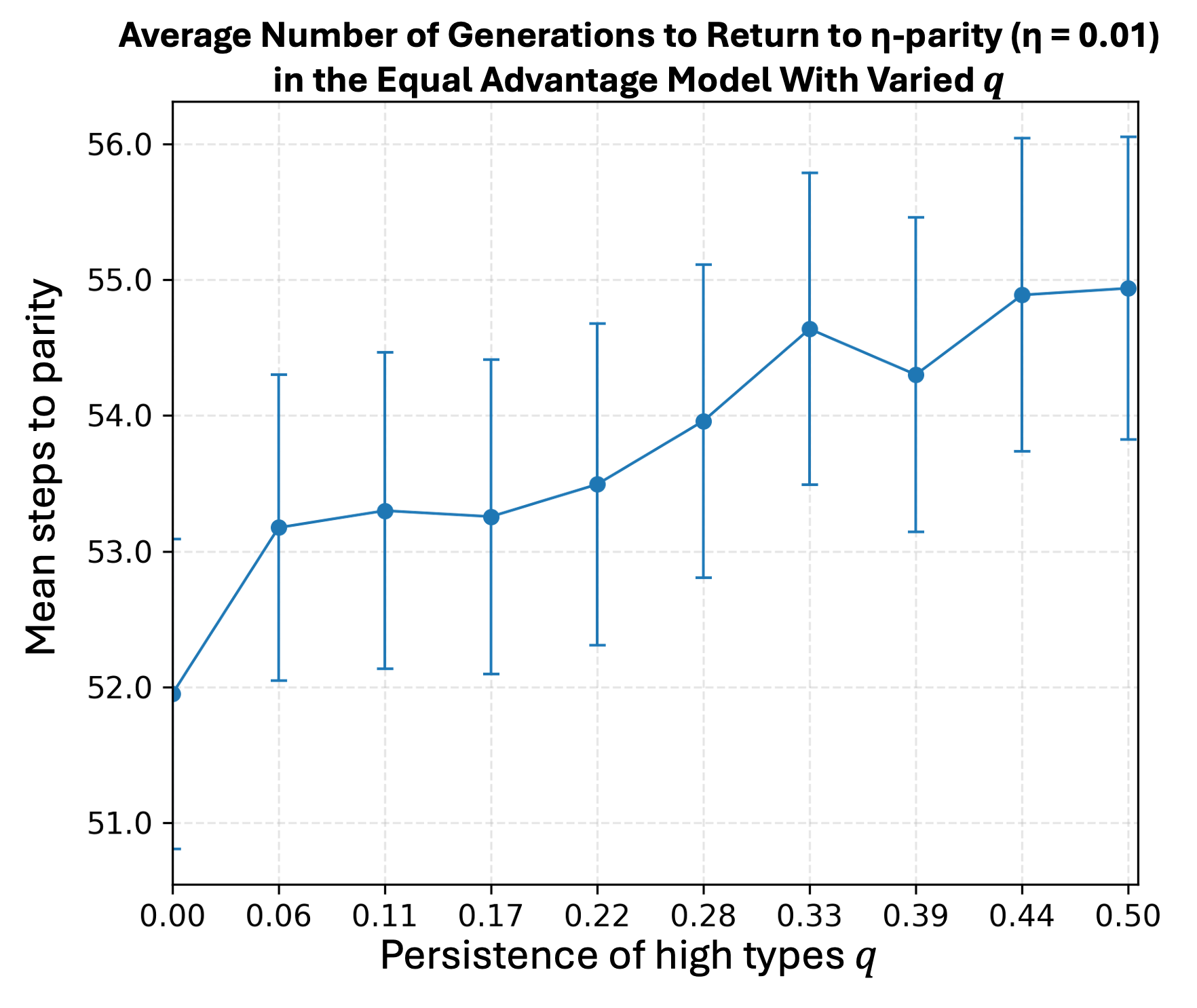}
        \caption{Dependence on $q$, with fixed $p = 0.95$}
        \label{fig:parity_q}
    \end{subfigure}
    \hfill
    \begin{subfigure}[b]{0.48\textwidth}
        \centering
        \includegraphics[height=6cm]{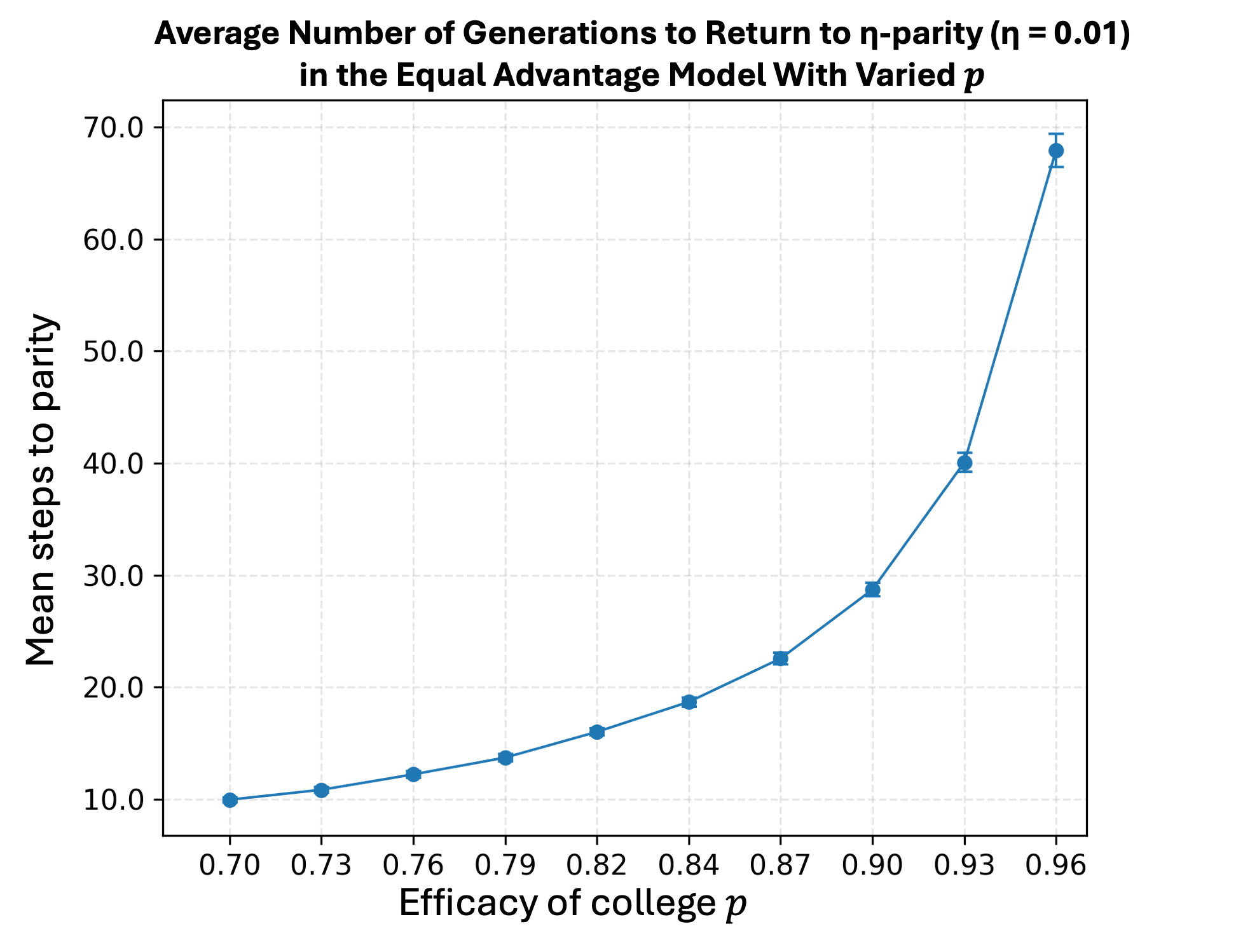}
        \caption{Dependence on $p$, with fixed $q = 0.40$}
        \label{fig:parity_p}
    \end{subfigure}
    \caption{Plots showing how the mean time to $\eta$-parity depends on $q$ (left) and $p$ (right) for a target $\eta$ of $0.01$ and a starting point $(X_A(0), X_B(0) = (0.9, 0.1))$ with an initial separation of $0.8$. Other parameters of interest: $\alpha = 0.30, N = 2000$. }
    \label{fig:return_to_parity_theory_comparison}
\end{figure*}

For a more comprehensive overview of how the mean time to parity depends on $\alpha$, $p$, $q$ for different combinations of starting points across regimes, refer to the heatmap in Figure~\ref{fig:ea_parity_heatmaps}.

\begin{figure*}[!ht]
  \centering
    \includegraphics[width=\linewidth]{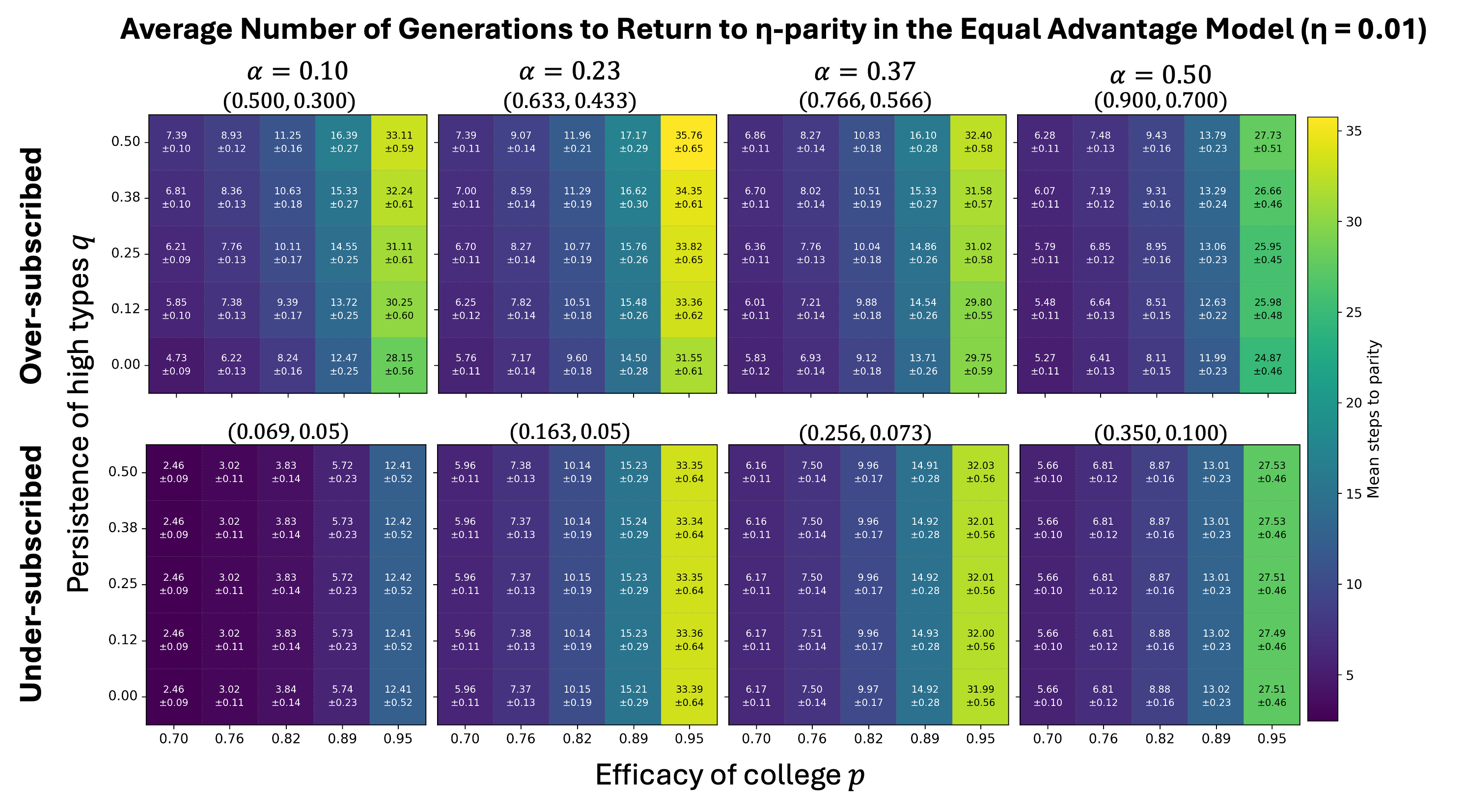}
      \caption{Average time (in generations) the system (under the EA model) needs to reach $\eta$-parity as a function of the college success probability $p$, probability of persistence of high types $q$, and capacity ~$\alpha$. The numbers in parenthesis above each heatmap shows the starting proportions of high types in the two groups i.e., $(X_A(0), X_B(0))$. We have $N=2000$ for all experiments. Each cell shows the mean time to parity over 100 runs; darker colors indicate faster return to parity. }
      \label{fig:ea_parity_heatmaps}
\end{figure*}

\subsection{Comparing Upper bounds provided by Theorem~\ref{thm:time_parity} to Numerical Experiments}
We compare how our theoretical high-probability bounds on the time to $\eta$-parity fare against empirical estimates of the same. We study this jointly as a function of the college success probability $p$ and the capacity parameter $\alpha$. We fix $N = 65,000$ (large enough according to the requirements of Theorem~\ref{thm:time_parity}), $q = 0.40$ and initialize the system at $(X_A(0), X_B(0)) = (0.9, 0.1)$, i.e., with an initial separation $|\Delta (0) = 0.8|$. Our target $\eta$ is $0.05$ and our goal is to compute the time to $\eta$-parity (empirically) with high confidence ($95\%$). Accordingly, we choose $\omega = 0.05$ to find upper bound $T_{\eta}$ with high probability (at least $0.95$) from Theorem~\ref{thm:time_parity}. We compute the ratio $r$ ($r \geq 1$) to see how loose our upper bound is. Figure~\ref{fig:thm3_ratios} plots these ratios for different $(\alpha, p)$ combinations in the form of a heatmap where lighter colors indicate smaller ratios. We find that at high $p$, our theoretical bound is consistently more conservative. But in general, $r \leq 5$.  

\begin{figure*}[ht]
  \centering
    \includegraphics[width=0.8\linewidth]{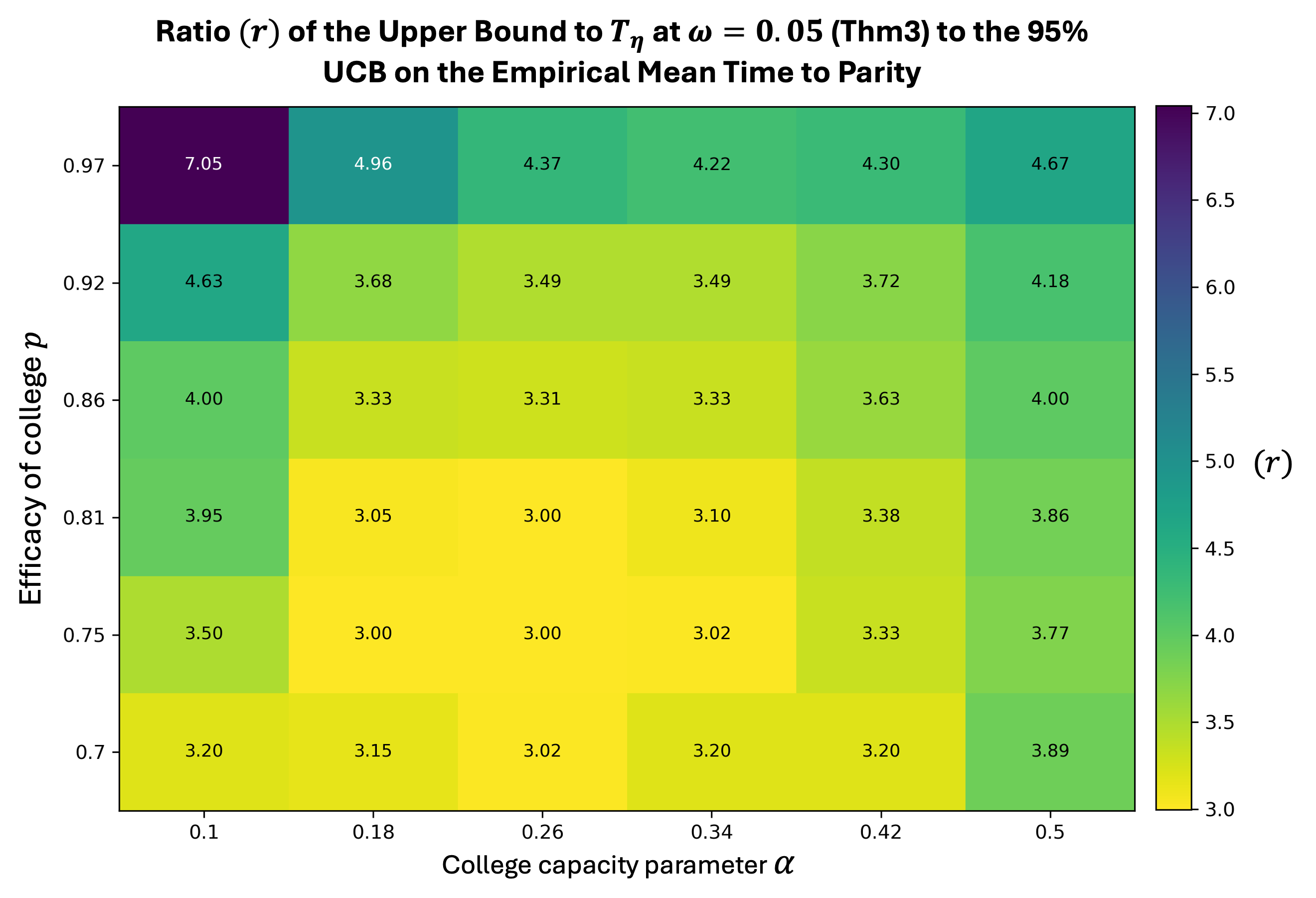}
      \caption{Heatmap showing how the upper bound in Theorem~\ref{thm:time_parity} compares to the empirical time to parity calculations. Experiment settings: $q=0.40$, $(X_A(0),X_B(0))=(0.90,0.10)$, $|\Delta(0)|=0.80$, $\eta=0.05$, $\omega=0.05$, $N = 65000$. The value corresponding to each $(\alpha, p)$ combination is computed using $100$ independent runs. Lighter colors indicate smaller ratios which imply smaller gap between theoretical predictions and empirical values.}
      \label{fig:thm3_ratios}
\end{figure*}

\subsection{Affinity Advantage Model: Results for $q > 0$ setting}
We now extend our analysis for the Affinity Advantage model with $q = 0$ to accommodate scenarios with $q > 0$. Note that the fixed point equations for the under-subscribed setting still remain unchanged, but for the over-subscribed regime, we have the following revised equations: 
\begin{align*}
    x_A &= \frac{x_A(2\alpha)}{x_A + x_B}\cdot p + \left(x_A - \frac{x_A (2\alpha)}{x_A + x_B} \right)\cdot (q+\epsilon) + (1-x_A)\cdot \epsilon,\\
    x_B &= \frac{x_A(2\alpha)}{x_A + x_B}\cdot p + \left(x_A - \frac{x_A (2\alpha)}{x_A + x_B} \right)\cdot q. 
\end{align*}
Observe that when $q  = 0$, we recover our fixed point equations in Section~\ref{sub:affinity_adv}. 

Now in Figure~\ref{fig:delta_eps_panels}, we plot the equilibrium separation $\Delta$ of the system as a function of the affinity advantage parameter $\epsilon$ for different values of $q$. We use the $q = 0$ setting as a benchmark for comparison. Firstly, we find that the same threshold $\tilde \epsilon$ still applies, even for $q > 0$. The system also retains the unique fixed point property for all $\epsilon$ when $q > 0$ (starting points in both over-subscribed and under-subscribed regimes converge to the same equilibrium, given $\epsilon$, as we see in Figures~\ref{fig:delta_eps_under} and \ref{fig:delta_eps_over}). In the sub-$\tilde \epsilon$ regime, the system attains the same equilibrium as the $q = 0$ setting for $q > 0$ --- this was expected because the fixed point equations for the under-subscribed regime are independent of $q$ and therefore, any equilibrium in the under-subscribed regime should also be independent of $q$. When $\epsilon \geq \tilde \epsilon$, however, we see that the system achieves an equilibrium separation $\Delta_{q > 0}$ which depends on $q$ and is given by: 
\[
     \Delta_{q > 0} = \frac{2\alpha (p-q-\epsilon)+\epsilon}{1-q} > \Delta_{q = 0}, 
\]
and is increasing in $q$. These experiments provide conclusive evidence that: i) the $q = 0$ setting captures all the properties/intricacies of the general $q > 0$ setting, ii) when $q > 0$, the equilibrium separation between groups can only be worse. 

\begin{figure*}[!ht]
    \centering
    \begin{subfigure}[b]{0.48\textwidth}
        \centering
        \includegraphics[width=\textwidth]{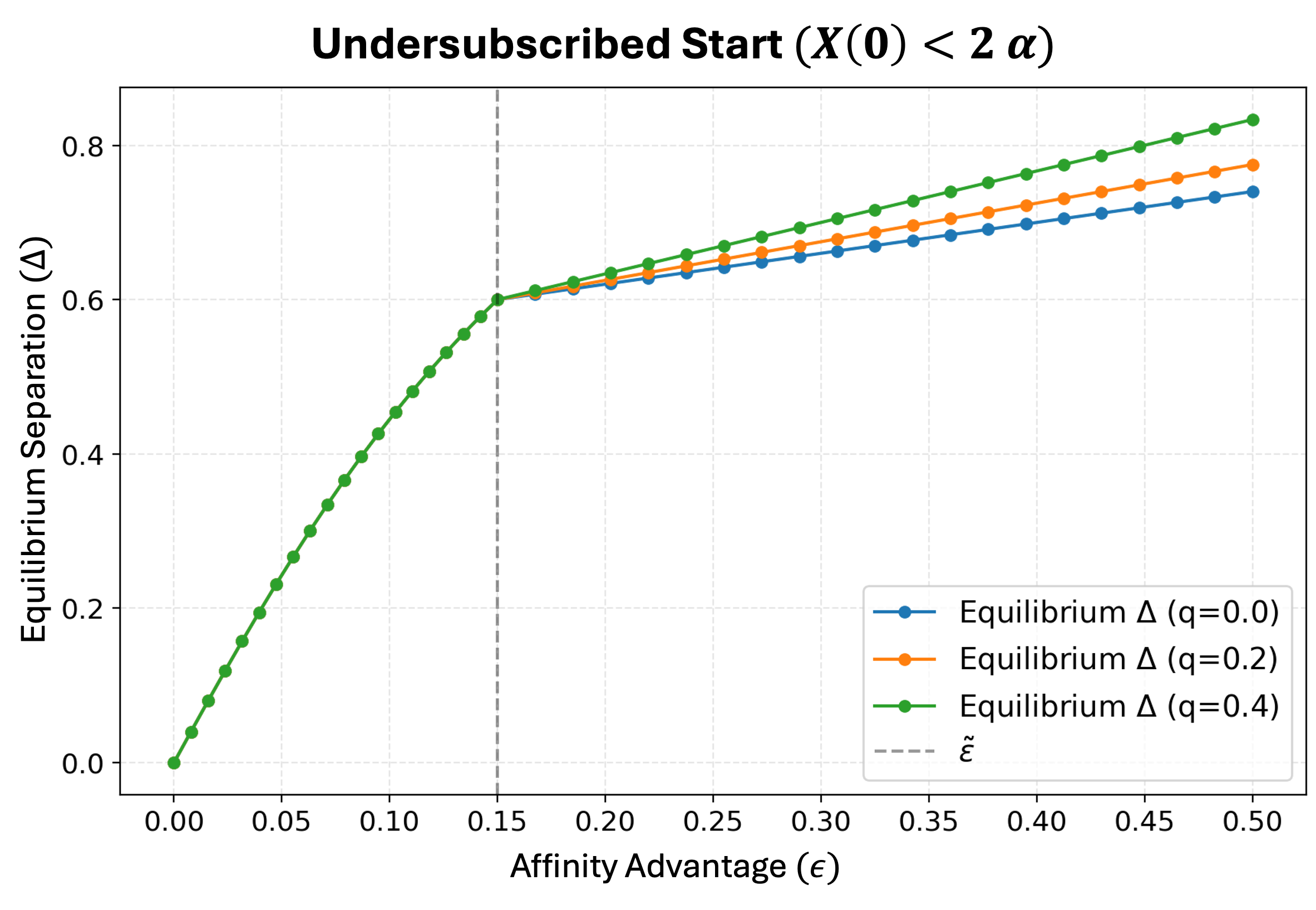}
        \caption{}
        \label{fig:delta_eps_under}
    \end{subfigure}
    \hfill
    \begin{subfigure}[b]{0.48\textwidth}
        \centering
        \includegraphics[width=\textwidth]{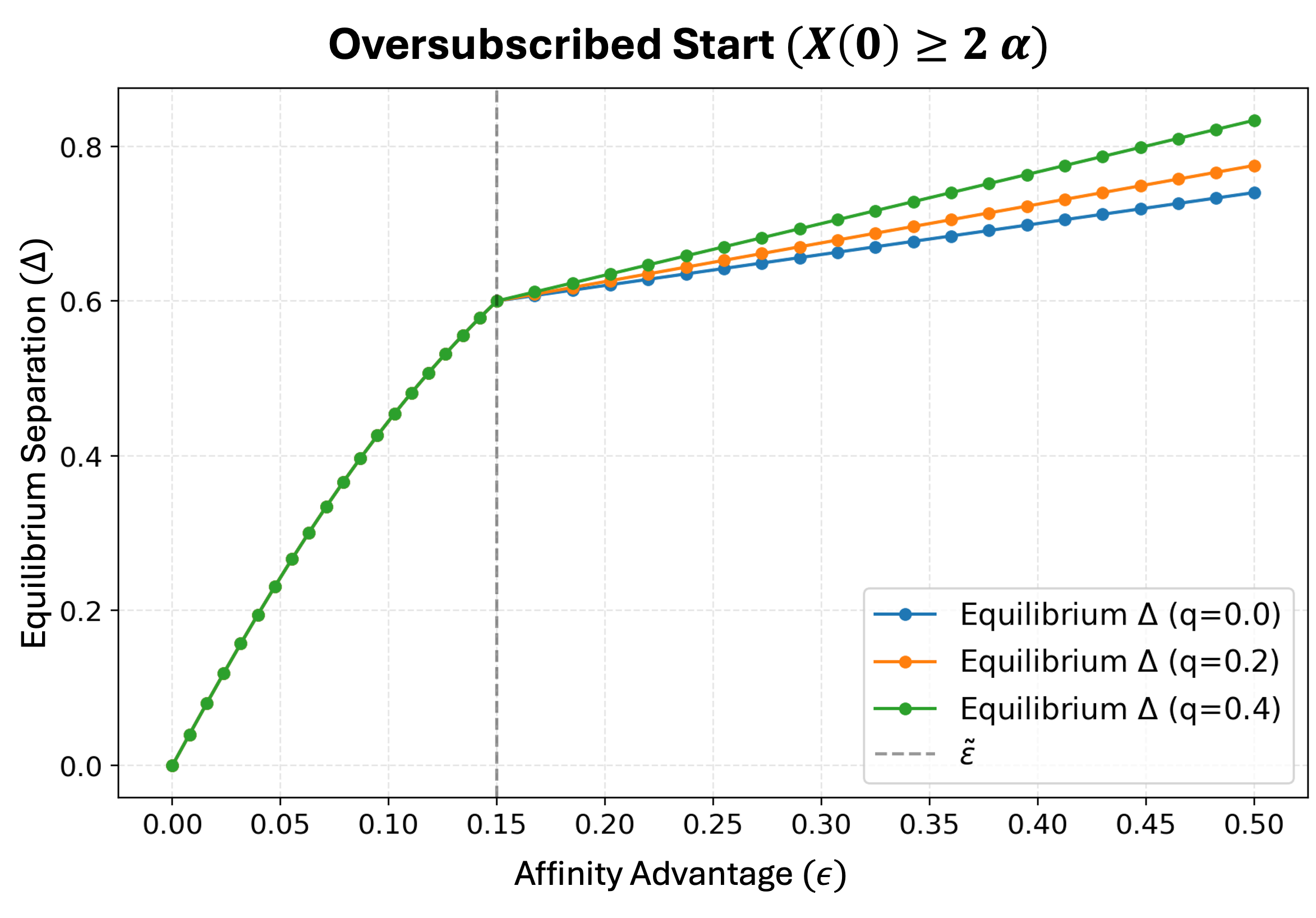}
        \caption{}
        \label{fig:delta_eps_over}
    \end{subfigure}
    \caption{Figure showing equilibrium separation as a function of the affinity advantage when $q > 0$. Parameter combination: $\alpha=0.30$, $p=0.90$. We try $q$ values in $\{0.0, 0.2, 0.4\}$ with $q = 0.0$ as a benchmark. Left and right represent under-subscribed $(X_A(0),X_B(0))=(0.20,0.19)$ and over-subscribed start $(0.5005,0.4995)$ respectively. When $q > 0$, the equilibrium separation between groups can only be worse. }
    \label{fig:delta_eps_panels}
\end{figure*}

\subsection{Affinity Advantage Model: More numerical experiments for $q > 0$}
In Figure~\ref{fig:aa_stable_heatmap}, we provide comprehensive numerical results that capture the long-run separation between groups under affinity advantage. We fix $N = 1000$,$p = 0.90$ and $q = 0.40$, but vary the capacity parameter $\alpha$ and the affinity advantage parameter $\epsilon$. The dashed line in both subplots indicates the threshold $\tilde \epsilon$ which depends on $\alpha$. Each grid represents an $(\alpha, \epsilon)$ pair for which we empirically obtain the long-run group separation averaged over $1000$ independent runs. Lighter grids indicate larger group separation in the long-run(left) or larger group share for the leading group (right). 

We observe that at low capacities, even small $\epsilon$ can lead to very large group separations in the long-run --- this indicates that network effects can exacerbate (already significant) separation effects generated by stochasticity alone when capacity is limited (observed earlier in Figure~\ref{fig:separation_N}). Therefore, investing in more capacity can help to mitigate some of these adverse outcomes.

\begin{figure}[ht]
  \centering
    \includegraphics[width=\columnwidth]{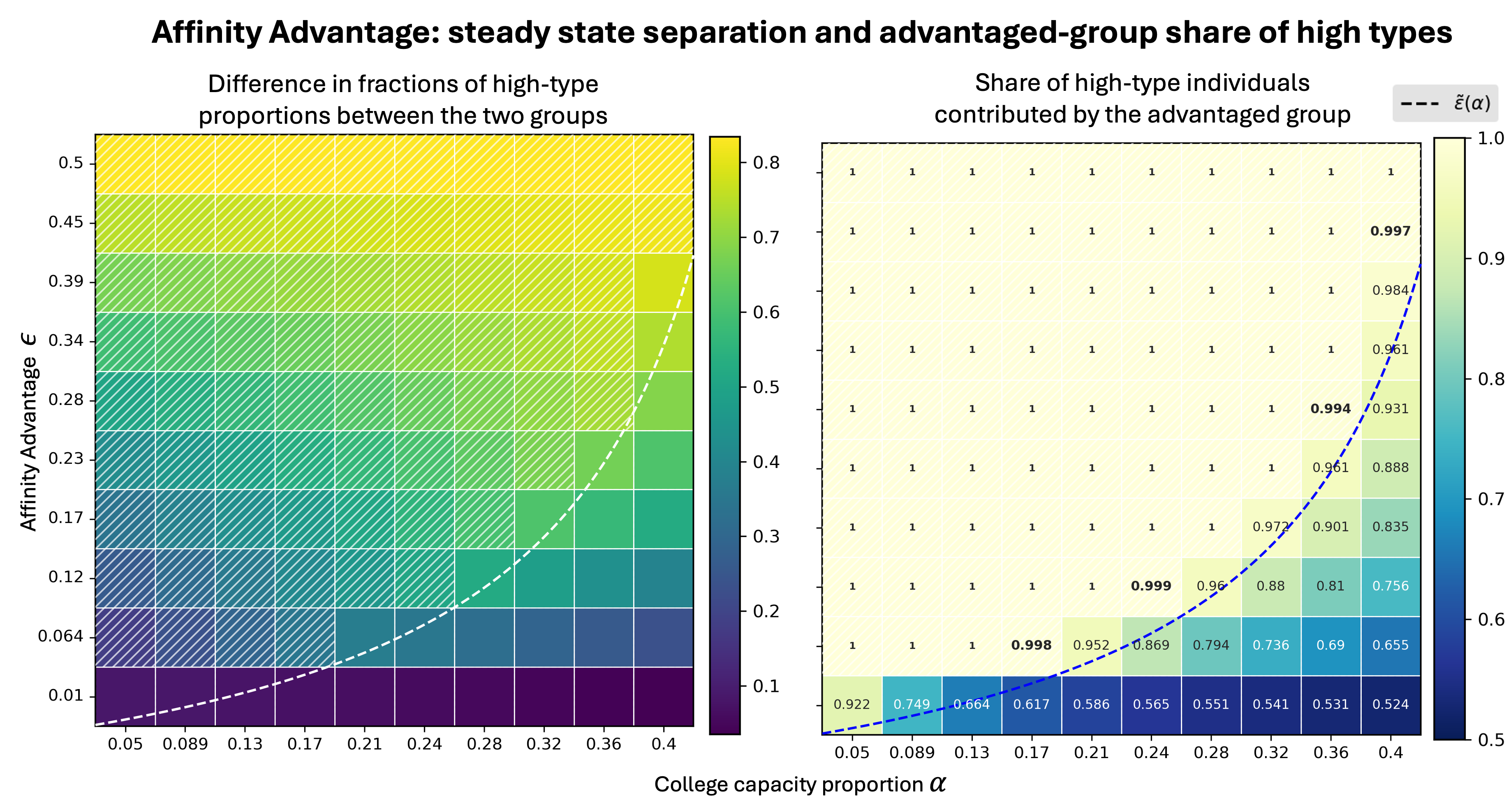}
    \caption{Outcomes for the Affinity Advantage model, showing the long-run separation (left) and the leading group's share of high-types (right). The dashed curve marks the threshold $\tilde{\epsilon}(\alpha) = 2\alpha(1-p)/(1-2\alpha)$ for extreme dominance (hatched region). Parameter combination: $p=0.90, q=0.40, N=1000, (X_A(0), X_B(0)) = (0.20, 0.20)$. At low capacities, long-run separation effects are exacerbated. }
    \label{fig:aa_stable_heatmap}
\end{figure}

\subsection{Richer Model: Sensitivity Analysis}
\label{app:sensitivity}

To ensure that our findings regarding group separation in the richer admissions models are not artifacts of specific parameter choices, we examined the system's behavior across a broader range of settings. We specifically investigated the impact of two key parameters: college efficacy ($\mu_I$) and generational noise ($\sigma_\lambda$). These results are shown in Figure ~\ref{fig:sensitivity_analysis_combined}. 

\begin{figure*}[ht]
  \centering
  \begin{subfigure}[t]{0.48\textwidth} 
    \centering
    \includegraphics[width=\textwidth]{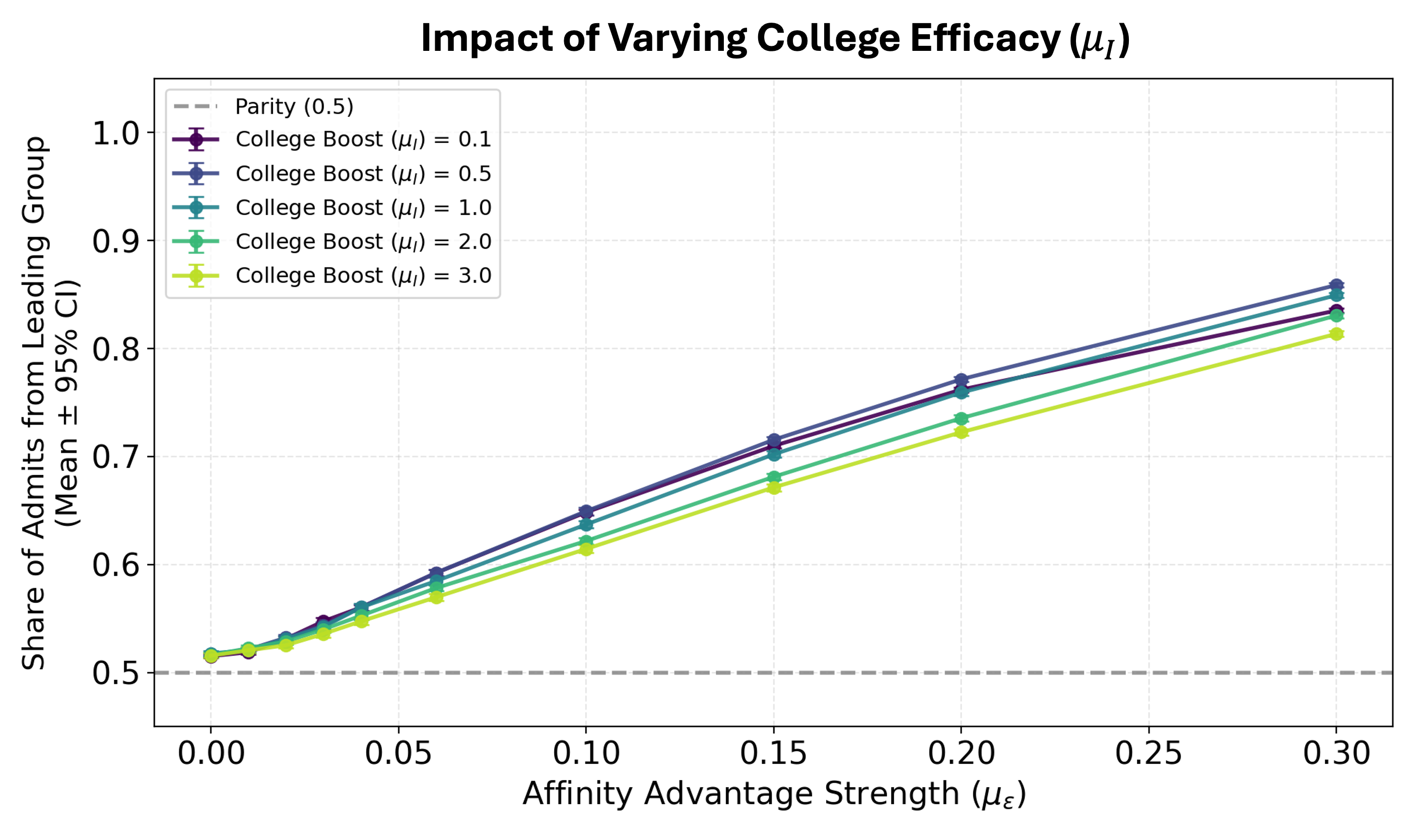} 
    \caption{ }
    \label{subfig:sens_muI}
  \end{subfigure}
  \hfill 
  \begin{subfigure}[t]{0.48\textwidth}
    \centering
    \includegraphics[width=\textwidth]{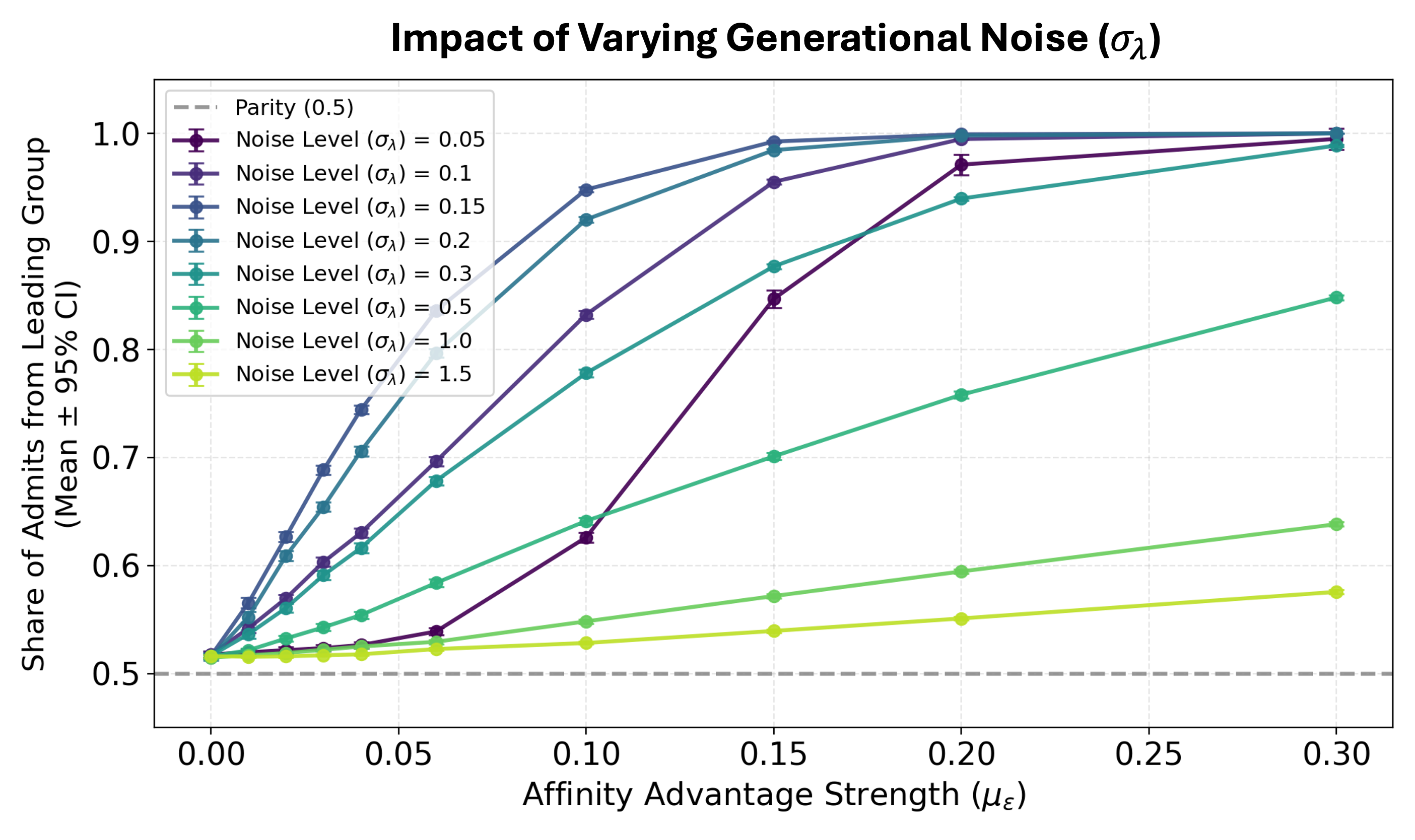}
    \caption{ }
    \label{subfig:sens_sigma}
  \end{subfigure}

  \caption{Sensitivity analysis of long-term admission outcomes. The figures show the share of admits from the leading group (mean $\pm$ 95\% CI over 100 runs) plotted against affinity advantage $\mu_\epsilon$. These experiments confirm that long-term separation is a robust feature of the dynamics and not an artifact of specific parameter choices for college boost or generational transmission noise.}
  \label{fig:sensitivity_analysis_combined}
\end{figure*}

\paragraph{Varying College Efficacy} First, we varied the college efficacy parameter in the range $\mu_I \in [0.1, 3.0]$. As shown in Figure~\ref{subfig:sens_muI}, increasing the meritocratic boost ($\mu_I$) slightly decreases the equilibrium separation. This effect is likely due to the meritocratic boost dampening the relative impact of the affinity advantage noise. However, notably, even at higher levels of $\mu_I$, the separation is not entirely eliminated, confirming the persistence of the advantage.

\paragraph{Varying Generational Noise} Second, we varied the generational noise parameter in the range $\sigma_\lambda \in [0.05, 1.5]$. Figure~\ref{subfig:sens_sigma} yields a non-monotonic relationship between generational noise and separation: separation is lower at the extremes (attributable to system rigidity at low noise and the dilution of advantage by noise at high levels) but peaks at intermediate levels ($\sigma_\lambda \approx 0.15$). These results confirm that long-term separation is a robust feature of the dynamics across a wide parameter space.

\end{document}